\documentclass[acmsmall,10pt, nonacm]{acmart}
\acmConference[PL'18]{ACM SIGPLAN Conference on Programming Languages}{January 01--03, 2018}{New York, NY, USA}
\acmYear{2018}
\acmISBN{} 
\acmDOI{} 
\startPage{1}

\setcopyright{none}

\usepackage{booktabs}   
\usepackage{subcaption} 
\usepackage{amsmath}
\usepackage{anyfontsize}
\usepackage[all,cmtip]{xy}
\usepackage{bcprules}
\usepackage{pifont}
\usepackage[nameinlink,noabbrev,capitalise]{cleveref}
\usepackage{csquotes}
\usepackage{booktabs}
\usepackage{microtype}
\usepackage{stmaryrd}
\usepackage{graphicx}
\usepackage{enumitem}
\usepackage{mathtools}
\usepackage{courier}
\usepackage{bussproofs}
\usepackage{tikz-cd}
\usepackage{diagbox}
\usepackage{afterpage}
\usepackage[normalem]{ulem}
\usepackage{dashbox}
\usepackage{marvosym}
\usepackage{cmll}

\usepackage[mathscr]{euscript}
\DeclareMathAlphabet\mathbfscr{U}{eus}{b}{n}

\usepackage{xcolor}
\usepackage{colortbl}
\usepackage[]{todonotes}

\usepackage{syntax}

\newcommand{\greybox}[1]{\colorbox{black!20}{\ensuremath{#1}}}

\newcommand{\cmark}{\ding{51}}%
\newcommand{\xmark}{\ding{55}}%

\newcommand{\sep}{\quad | \quad}

\newcommand{\OK}{\textsc{Ok}}

\newcommand{\CU}{\textbf{CU}}
\newcommand{\DualCalculus}{\textbf{CD}}
\newcommand{\lambdamutildemu}{$\bar{\lambda} \mu \tilde\mu$}
\newcommand*{\derefunc}{de/re\-func\-tion\-ali\-za\-tion}
\newcommand*{\Derefunc}{De/Re\-func\-tion\-ali\-za\-tion}
\newcommand*{\cbxtocby}{evaluation order switch}

\newcommand{\etaeq}{\ensuremath{=_{\eta}}}
\newcommand{\idmorphism}[1]{\ensuremath{\text{id}_{#1}}}

\newcommand{\data}{\ensuremath{\mathbf{data}}}
\newcommand{\codata}{\ensuremath{\mathbf{codata}}}
\newcommand{\match}[1]{\mathbf{match}_{#1}}

\newcommand{\case}[3]{\texttt{#1}(#2) \Rightarrow #3}

\newcommand{\cbv}{\ensuremath{\mathbf{cbv}}}
\newcommand{\cbn}{\ensuremath{\mathbf{cbn}}}

\newcommand*{\flip}[1]{\hat{#1}}

\newcommand{\mkCmd}{\mathrel{\gg}}
\newcommand{\Done}{\ensuremath{\mathbf{Done}}}

\newcommand{\muprdabs}[3]{\ensuremath{\mu(#1 \con #2).#3}}
\newcommand{\muconabs}[3]{\ensuremath{\mu(#1 \prd #2).#3}}
\newcommand{\muqabs}[3]{\ensuremath{\mu(#1 \overset{o}{:} #2).#3}}

\newcommand{\valfun}[1]{\ensuremath{\mathbf{val}(#1)}}
\newcommand{\cntfun}[1]{\ensuremath{\mathbf{cnt}(#1)}}

\newcommand{\wrap}[2][]{\,\overset{\textit{\tiny#2}}{\colon}#1\,}
\newcommand{\prd}{\wrap{\raisebox{0pt}{\ensuremath{\rule{0pt}{3pt}\smash{\mathbf{prd}}}}}}
\newcommand{\prdcon}{\wrap{\raisebox{0pt}{\ensuremath{\rule{0pt}{3pt}\smash{o}}}}}
\newcommand{\con}{\wrap{\rule{0pt}{1pt}\ensuremath{\mathbf{con}}}}
\newcommand{\cmd}{\wrap{\ensuremath{\mathbf{cmd}}}}
\newcommand{\substitutable}{\textsc{Subst}}

\newcommand{\reducesto}{\mathrel{\triangleright}}
\newcommand{\reducestoinv}{\mathrel{\triangleleft}}

\newcommand*{\xfunLetter}{\mathcal{C}}
\newcommand{\xfunctionalize}{\ensuremath{\xfunLetter_p}}
\newcommand{\xfunctionalizeinv}{\ensuremath{\xfunLetter_{\flip{p}}}}
\newcommand{\refunctionalize}{\ensuremath{\xfunLetter_{\codata}}}
\newcommand{\defunctionalize}{\ensuremath{\xfunLetter_{\data}}}

\newcommand*{\fullXfunLetter}{\mathcal{F}}
\newcommand{\fxfunctionalize}{\ensuremath{\fullXfunLetter_p}}
\newcommand{\frefunctionalize}{\ensuremath{\fullXfunLetter_{\codata}}}
\newcommand{\fdefunctionalize}{\ensuremath{\fullXfunLetter_{\data}}}

\newcommand{\cbxshiftCommon}[2]{\ensuremath{\uparrow^{#1}\!#2}}
\newcommand{\cbvshift}[1]{\cbxshiftCommon{\cbv}{#1}}
\newcommand{\cbnshift}[1]{\cbxshiftCommon{\cbn}{#1}}
\newcommand{\cbxshift}[1]{\cbxshiftCommon{s}{#1}}
\newcommand{\cbxshiftinv}[1]{\cbxshiftCommon{\flip{s}}{#1}}

\newcommand{\cbvterm}[1]{\ensuremath{\texttt{CBV}(#1)}}
\newcommand{\cbnterm}[1]{\ensuremath{\texttt{CBN}(#1)}}

\newcommand{\evaltransnametemplate}[2]{\mathcal{T}^{#1}_{#2}}
\newcommand{\evaltranstemplate}[3]{\ensuremath{\evaltransnametemplate{#1}{#2}\!(#3)}}
\newcommand{\evaltransinv}[1]{\evaltranstemplate{T}{\flip{s}}{#1}}

\newcommand{\evaltrans}[1]{\evaltranstemplate{T}{s}{#1}}
\newcommand{\evaltransname}{\ensuremath{\evaltransnametemplate{T}{s}}}

\renewcommand{\S}{\ensuremath{\mathcal{S}}}

\newcommand{\doubletrans}[1]{\evaltransinv{\evaltrans{#1}}}

\newcommand*{\cbxLetter}{\mathcal{T}}
\newcommand{\cbvtocbn}{\ensuremath{\cbxLetter_{\cbn}}}

\newcommand{\cbntocbv}{\ensuremath{\cbxLetter_{\cbv}}}

\newcommand{\X}{\ensuremath{\mathcal{X}}}
\newcommand{\Y}{\ensuremath{\mathcal{Y}}}
\newcommand{\N}{\ensuremath{\mathbb{N}}}

\begin{document}

\title{Data-Codata Symmetry and its Interaction with Evaluation Order}
\subtitle{}

\author{David Binder}
\orcid{nnnn-nnnn-nnnn-nnnn}             
\affiliation{
  \department{Department of Computer Science}              
  \institution{University of Tübingen}            
  \streetaddress{Sand 14}
  \city{Tübingen}
  \postcode{72076}
  \country{Germany}
}
\email{david.binder@uni-tuebingen.de}          

\author{Julian Jabs}
\orcid{nnnn-nnnn-nnnn-nnnn}             
\affiliation{
  \department{Department of Computer Science}              
  \institution{University of Tübingen}            
  \streetaddress{Sand 14}
  \city{Tübingen}
  \postcode{72076}
  \country{Germany}
}
\email{julian.jabs@uni-tuebingen.de}          

\author{Ingo Skupin}
\orcid{nnnn-nnnn-nnnn-nnnn}             
\affiliation{
  \department{Department of Computer Science}              
  \institution{University of Tübingen}            
  \streetaddress{Sand 14}
  \city{Tübingen}
  \postcode{72076}
  \country{Germany}
}
\email{skupin@informatik.uni-tuebingen.de}          

\author{Klaus Ostermann}
\orcid{nnnn-nnnn-nnnn-nnnn}             
\affiliation{
  \department{Department of Computer Science}              
  \institution{University of Tübingen}            
  \streetaddress{Sand 14}
  \city{Tübingen}
  \postcode{72076}
  \country{Germany}
}
\email{klaus.ostermann@uni-tuebingen.de}          

\ccsdesc[500]{Theory of computation~Operational semantics}
\ccsdesc[500]{Theory of computation~Abstract machines}
\ccsdesc[500]{Software and its engineering~Multiparadigm languages}
\ccsdesc[500]{Software and its engineering~Abstract data types}

\keywords{Data and codata, Operational semantics, Defunctionalization, Refunctionalization, Duality}

\begin{abstract}
  Data types and codata types are, as the names suggest, often seen as duals of each other. However, most
  programming languages do not support both of them in their full generality, or if they do, they are still
  seen as distinct constructs with separately defined type-checking, compilation, etc.

  \citeauthor{uroboro2015} were the first to propose variants of two standard program transformations, 
  defunctionalization and refunctionalization, as a test to gauge and improve the symmetry between data and codata types.
  However, in previous works, codata and data were still seen as separately defined language constructs, with
  defunctionalization and refunctionalization being defined as similar but separate algorithms. These works also 
  glossed over  interactions between the aforementioned transformations and evaluation order, which leads 
  to a loss of desirable $\eta$ expansion equalities.

  We argue that the failure of complete symmetry is due to the inherent asymmetry of natural deduction as the logical
  foundation of the language design. Natural deduction is asymmetric in that its focus is on \emph{producers} (proofs)
  of types, whereas consumers (contexts, continuations, refutations) have a second-class status. Inspired by existing
  sequent-calculus-based language designs, we present the first language design that is fully symmetric in that the issues of polarity (data type vs codata types) and
  evaluation order (call-by-value vs call-by-name) are untangled and become independent attributes of a single form
  of type declaration. Both attributes, polarity and evaluation order, can be changed independently by \emph{one} algorithm each.
  In particular, defunctionalization and refunctionalization are now \emph{one} algorithm.
  Evaluation order can be defined and changed individually for each type, independently from polarity. By allowing only certain combinations of
  evaluation order and polarity, the aforementioned $\eta$ laws can be restored.
\end{abstract}

\maketitle

\section{Introduction}
\label{sec:intro}

Data types, codata types and their associated cousins of induction and coinduction have a
rich history in the programming language literature. Data types, in the form of algebraic
data types and pattern matching, are a staple of functional programming. Codata types, on the other
hand, are traditionally often only supported in the degenerate form of $\lambda$-abstractions (the special
case of a codata type with a single \emph{apply} destructor)%
\footnote{Objects as in object-oriented languages can be seen as a more general form of codata \cite{cook09understanding}:
an object is an instance of some class which implements an interface which specifies destructors (methods).
Other features of OO languages, such as inheritance, do not readily fit with this connection though.}.

More recently, the support for codata types in programming languages has grown, especially after
the discovery of copattern matching \cite{abel13copatterns} as an efficient notation; for instance,
in OCaml \cite{10.1145/3131851.3131869} and Agda \cite{cockxabel2020}. However, even in those
languages that do support codata types, the apparent symmetry between data types and codata
types is not used: Type checking, implementation etc.~is separate; transformations between
data types and codata types are manual.

The aim of this work is to improve the symmetry between data and codata.
There are three frameworks in which they can be compared. First, from the
perspective of universal algebra and category theory, the two can be related via their respective
semantics as initial algebras and final coalgebras. Second, from the point of the Curry-Howard isomorphism,
the two can be related via the usual dualities from logic: The sum data type is dual to product codata type
in the same way that disjunction is dual to conjunction, and so forth.

In this work, we are using a third method that was first proposed by \citet{uroboro2015} and
later elaborated by \citet{ostermann2018dualizing} and \citet{Uroboro2019}, namely
to relate data and codata types by generalizations of two well-known \emph{global} program transformations,
defunctionalization \cite{reynolds72definitional,danvy01defunctionalization} and refunctionalization \cite{danvy09refunctionalization}.
Refunctionalization turns a data type into a codata type by turning every constructor of the data type
into a function containing a copattern match on the codata type, and turning every pattern match
on the data type into a destructor invocation of the codata type. For instance, a data type for natural numbers
with the standard {\tt Suc} and {\tt Zero} constructors and a program that contains a single pattern match
(let's call it {\tt isZero}) that returns {\tt True} iff the number is {\tt Zero} is turned into a codata
type with a single destructor {\tt isZero} and two functions {\tt Suc} and {\tt Zero} that
copattern-match on the codata type to produce {\tt False} and {\tt True}, respectively, when the {\tt isZero}
destructor is invoked. It is essential that this transformation is global; the destructors are
determined by the set of pattern matches that can be found in the program. Defunctionalization is defined
in a similar way in the other direction.

From a language design perspective, the interesting point about these transformations is that they provide
a novel and tangible way to compare and relate the data and codata type design of the language. This paper
aims to improve on the state of the art in two regards:

\subsection*{First problem: Lack of Symmetry}
The first problem with existing language designs for data and codata types is a lack of symmetry.
The aforementioned \derefunc\ transformations can serve as a tool to evaluate symmetry: A failure to
define these transformations as total, semantics-preserving and mutually inverse functions indicate an asymmetry.
For example, while \citet{pottier06polymorphic} figured out that generalized algebraic data types (GADTs) are necessary
in order to correctly defunctionalize polymorphic functions, the requirement for total and inverse de- and refunctionalization
made it clear that the dual feature of generalized algebraic codata types (GAcoDTs) are also necessary \cite{ostermann2018dualizing}, and
what they should look like.

However, even when the design is symmetric in that sense, there is another glaring asymmetry that has been the main motivation
for this paper: The definitions of data and codata types as well as the transformations between them look very similar, but not similar
enough to have just one generic definition of a type (with data and codata being just two modes of using them) and just one transformation.

The underlying cause which we identified is the asymmetry between producers and consumers.
Programming languages whose design is based on natural deduction rules only represent producers (expressions) of a type as a first class citizen.
That is, there is only one form of typing judgement $\Gamma \vdash t : T$ for typing (producer) terms.
This bias in the design of programming languages, which is reflected in their formalizations, manifests itself in various ways:
\begin{itemize}
  \item A constructor such as ``Zero'' can be typed on its own, whereas a destructor such as ``isZero'' must be typed together with its application, e.g., ``3.isZero''.
  \item Pattern matches have a separate ``return type'' in addition to the type of the argument they pattern match on, while copattern matches lack such a separate return type.
  \item Constructors of a data type and observations of a codata type have a different shape, observations have an additional ``return type'' of the observation.
\end{itemize}
Sequent calculus is a more symmetric alternative to natural deduction in which proofs\slash{}pro\-ducers and refutations\slash{}con\-sumers are defined in a completely symmetric way.
Inspired by previous work on sequent-calculus-based languages
\cite{zeilberger2008unity,curien00duality,downen14duality,downen2019compiling,wadler03call}, we have been able to achieve our:

\begin{enumerate}
  \item[] \textbf{First contribution:} We present a novel language in which data and codata types are \emph{fully} symmetric and in fact just modes of one type definition construct.
   Defunctionalization and refunctionalization are \emph{one} algorithm.
\end{enumerate}
Programs in sequent calculus style look very different from ``normal'' functional programs (``the plumbing is exposed on the outside'' \cite{wadler03call}); it is not obvious whether
programmers want or should directly program in such languages (maybe aided by the ability to macro-expand normal expression-style syntax into
sequent calculus style) or whether they are better suited as intermediate languages. In any case, this paper should be considered
a language design experiment whose aim is to achieve the contribution just described, not necessarily to present a language in which
programmers program directly.

\subsection*{Second problem: Evaluation Order}
The second problem concerns the interaction between data and codata types, \derefunc, and evaluation order.
In languages that support both data and codata types, it is desirable to have fine-grained control over evaluation order
instead of prescribing a global fixed evaluation order:
\begin{itemize}
    \item Programmers want to use laziness for easier problem decomposition and composability of algorithms, and strictness for easier reasoning about time and space complexity.
    \item Ad hoc solutions to a globally fixed evaluation order like the ``seq'' expression in Haskell weaken the valid reasoning principles available to programmers \cite{johann2004freeTheoremsSeq}.
    \item The validity of $\eta$-rules depends both on the evaluation order of the language and the polarity of the type.
\end{itemize}
Since the last point is central to our argument, we illustrate it with two examples from the lambda calculus with pairs
(adapted from \citet{downen2019compiling}).

The $\eta$ rule for the function type (and for codata types more generally) is only valid under call-by-name evaluation order.
Consider, for example, the lambda term $(\lambda x.5)(\lambda x. \Omega x)$, which can be reduced to 5 under both the call-by-value and call-by-name evaluation order.
If we replace $\lambda x. \Omega x$ by its $\eta$-reduct $\Omega$, on the other hand, the resulting term diverges under call-by-value semantics, while it still reduces to 5 under call-by-name:

\begin{center}
  \begin{tikzcd}[row sep=-.6em,%
    /tikz/column 4/.append style={anchor=base west},%
  ]
    5 & & & 5\\
    & \ar[ul, swap, "\cbn"] \ar[dl, "\cbv"] (\lambda x.5)(\lambda x. \Omega x) \ar[r,equal, "\eta"]& (\lambda x.5) \Omega \ar[ur, "\cbn"] \ar[dr, "\cbv"'] &\\
    5 & & & \text{\emph{diverges}}
  \end{tikzcd}
\end{center}
The inverse situation holds for pairs (and data types more generally); their $\eta$-rules are only valid under call-by-value evaluation order.
The following two terms are $\eta$-equal.
But while the left term diverges under both call-by-value and call-by-name, the $\eta$-reduct on the right reduces to $5$ under call-by-name.

\begin{center}
  \begin{tikzcd}[row sep=-.6em,%
    /tikz/column 1/.append style={anchor=base east},%
    /tikz/column 4/.append style={anchor=base west},%
  ]
    \text{\emph{diverges}} & & & 5\\
    & \ar[ul, "\cbn"'] \ar[dl, "\cbv"]
        \mathbf{case}\ \Omega\  \mathbf{of}\ \{\ \langle x_1,x_2 \rangle \Rightarrow (\lambda x. 5) \langle x_1,x_2 \rangle\ \}
        \ar[r, equal, "\eta"]
      & (\lambda x. 5)\Omega
        \ar[ur, "\cbn"] \ar[dr, "\cbv"']
      &\\
      \text{\emph{diverges}} & & & \text{\emph{diverges}}
  \end{tikzcd}
\end{center}
These problems suggest that evaluation order should be based on types \cite{downen2019compiling}, together with so-called shifts (analogous to shifts in \cite{zeilberger2008unity}) to switch between different evaluation orders.
When relating data and codata types via \derefunc, however, these different evaluation orders and the shifts need to be taken
into account to preserve the semantics of the program. Previous work has sidestepped this problem by fixing one global evaluation order (either call-by-value or call-by-name).
This leads to our
\begin{enumerate}
  \item[] \textbf{Second contribution:} We give the first presentation on the interaction between data/codata types, type-based evaluation order, and \derefunc.
    We discuss, and solve, the problems which arise when we combine them.
\end{enumerate}

\subsection*{Relevance and Applications}
The main motivation for this paper is a conceptual one and stems from the desire to deepen the symmetry between data and codata types: i) to show that they can be expressed as modes of
just one type definition construct, ii) to demonstrate that variants of defunctionalization and refunctionalization can switch these modes and can be expressed as
a single algorithm, and iii) to analyze how type-based evaluation order to maintain desirable $\eta$-rules can be added to the language.

While applications of these conceptual ideas are not in the scope of this paper, we still want to mention that there are ample potential practical applications and benefits.
In general, the identification of dualities always comes with a two-for-the-price-of-one economy \cite{wadler03call}. The work presented here enables compilers, compiler optimizations,
virtual machines, static analyses, IDEs, proofs, and programmer tools to reuse the same code and ideas two times. As a concrete example, a compiler today might contain separate
optimizations that perform transformations like these:

$
\begin{array}{ccc}
(\lambda x. x) 5                   & \Rightarrow   &       5 \\
   \mathbf{match}\  \texttt{True}\ \mathbf{with}\ \{ \texttt{True} => s;\  \texttt{False} => t \}   & \Rightarrow   & s
\end{array}
$

\noindent In our system, both optimizations could be expressed as a single ``statically known con-/destructor'' rule. Even in the pedagogy of programming, a more symmetric treatment of data and codata
could lead to reuse of teaching concepts such as ``the structure of the program follows the structure of the (co)data'' \cite{gibbonsblog}.

Another potential application of this work is to use it as an intermediate language. \citet{downen2016sequent} have demonstrated that a language based on sequent calculus can be
an attractive intermediate language that lies ``somewhere in between direct and continuation-passing style'' and can, to some degree, combine the advantages of direct style
representation (such as: easy to specify rewrite rules, flexible evaluation order) with the advantages of CPS (such as: easy to express control flow). \citeauthor{downen2016sequent}'s
work was presented as an intermediate language for Haskell. What this work adds on top of the arguments presented by
\citeauthor{downen2016sequent} is the symmetric support for codata and type-based evaluation order, hence we envision that our language could be a common intermediate language for multiple different languages featuring different variants of data and
codata types and different evaluation order regimes. The intermediate language could make programs written in these different languages interact in a principled way, without
violating the invariants of the respective source languages.

\subsection*{Overview}
The rest of the article is structured as follows:
\begin{itemize}
  \item In \cref{sec:mainideas} we present our central contributions in an example-driven, informal style.
  \item In \cref{sec:formalization} we present the formalization of the CPS-fragment of the language.
    Since all programs have to be written with explicit control flow, different evaluation strategies cannot be observed.
  \item In \cref{sec:xtorization} we present the single algorithm which subsumes both defunctionalization and refunctionalization.
  \item In \cref{sec:evaluationorder} we remove the restriction to programs with explicit control flow by introducing the $\mu$ and $\tilde\mu$ constructs from \lambdamutildemu\ calculus.
    We discuss different global and type-based evaluation orders and their influence on the validity of $\eta$-equalities.
  \item In \cref{sec:xtorizationTwo} we discuss how to complement the transformations from \cref{sec:xtorization} with transformations which change the evaluation order of a given type.
  \item In \cref{sec:relatedwork} we discuss related work, and we conclude in \cref{sec:conclusion}.
\end{itemize}
The system presented in this paper has been formalized in Coq and \cref{thm:preservation,thm:progress,theorem:progsinverse,theorem:typeabilitypreservation,theorem:semantic} have been proven in Coq.
In this article we restrict ourselves to a simply-typed language, since this is sufficient to illustrate our central ideas.
We think that the generalization to a polymorphic variant of this calculus which supports generalized algebraic data types (GADTs) and their dual, GAcoDTs, does not pose substantial difficulties and expect this to be a straight-forward adaption of
the approach of \citet{ostermann2018dualizing}.

\section{Main Ideas}
\label{sec:mainideas}
In this section we will use simple examples to present the types and terms of the language, as well as the transformation algorithms.

\subsection{Programming with Symmetric Data and Codata}
Natural deduction, and the programming languages based on it, are biased towards proof.
A natural deduction derivation ending with a single formula at its root proves that formula.
Term assignment systems for natural deduction, such as the simply typed lambda calculus, therefore only have one typing judgement.
This typing judgement, usually written $\Gamma \vdash t : T$, types a term $t$ as a \emph{proof} of the type $T$.

Sequent calculus, on the other hand, is not biased towards proofs.
For example, there are not only sequent calculus derivations ending in $\vdash \phi$ which \emph{prove} a formula $\phi$, but also derivations ending in $\phi \vdash$ which \emph{refute} a formula $\phi$.
In the context of programming languages, proofs and refutations correspond to \emph{producers} and \emph{consumers} of a type.
For example, the producers of the type $\N$ are numbers $1,2,3\ldots$, whereas the consumers are \emph{continuations} expecting a natural number.
The single typing judgement $\Gamma \vdash t : T$ is replaced by two separate judgements $\Gamma \vdash p \prd T$ and $\Gamma \vdash c \con T$, one which types producers and one which types consumers.

There are exactly two kinds of types in our system, data types and codata types \citep{downen2019codata, hagino89codatatypes}, whose difference can be expressed in terms of \emph{canonical} producers and consumers.
Data types have canonical producers, which are called \emph{constructors}.
Canonical means that given an arbitrary producer, we know that it must have been built by one among a finite list of constructors.
Having only a finite list of constructors justifies the use of \emph{pattern matching} to build consumers of a data type.
Codata types, on the other hand, have canonical consumers, called \emph{destructors}; their producers are formed by \emph{copattern matching} \citep{abel13copatterns} on all destructors.

First-class support for both producers and consumers is necessary to make data and codata types completely symmetric.
For data types, constructors can be typed as producers and pattern matches as consumers.
Dually, destructors of a codata type can be typed as consumers, and copattern matches as producers.
Without a first-class representation of consumers we would have to treat data and codata types asymmetrically: both pattern matches and destructors need an additional return type.
We will now illustrate symmetric data and codata types with some simple examples.

The data type $\N$ is defined by two constructors \texttt{Zero} and \texttt{Suc}, the canonical producers of $\mathbb{N}$.
Since we distinguish producers from consumers, we have to explicitly mark the argument of the constructor \texttt{Suc} as a producer.
{\footnotesize
  \begin{align*}
    \data\ \mathbf{type}\ \N\ \{\ \texttt{Zero}\ ;\ \texttt{Suc}(x \prd \N)\ \}
  \end{align*}
}%
Using this definition, we can now form the following producer and consumer of $\N$:
{\footnotesize
  \begin{gather*}
    \texttt{Suc}(\texttt{Zero}) \prd \N \\
    \mathbf{match}_{\data}\ \N\ \{\ \texttt{Zero}\ \Rightarrow\ \ldots\ ;\ \case{Suc}{x \prd \N}{\ldots}\ \} \con \N
  \end{gather*}
}%
The definition of the consumer is still incomplete; we have not yet specified what should be inserted in the place of the two holes.
We do not have a ``result type'' which would determine the type of terms to put in the holes; rather, the category of terms to insert in these
places are called \emph{commands}, which we discuss next.

\emph{Commands}, sometimes also called ``statements'', are the syntactical category of reducible expressions; a closed command corresponds to the state of an abstract machine.
We consider only two types of command.
A logical command $p \gg c$ combines a producer $p$ and consumer $c$ of the same type.
Logical commands are evaluated using standard (co)pattern matching evaluation rules.
The second type of command is \Done\ which just terminates the program.
The addition of \Done\ is necessary since the Curry-Howard interpretation of a command is a contradiction \cite{zeilberger2008unity}.
The only way to write a closed command is therefore to write a looping program using unrestricted recursion, or to postulate the existence of a closed command, which we have done here.
We now present the first simple example of a complete program which reduces to \Done\ in a single step before terminating.
{\footnotesize
  \begin{align*}
    \texttt{Suc}(\texttt{Zero}) \gg \mathbf{match}_\data\ \N\ \{\ \texttt{Zero}\ \Rightarrow\ \Done\ ;\ \case{Suc}{x}{\Done} \}
  \end{align*}
}%

We now consider some examples of codata types.
Codata types are a powerful addition to statically typed programming languages, since they subsume a wide variety of different features.

The first example of a codata type, and the only codata type in many functional programming languages, is the function type.
Codata types allow the function type to be user-defined instead of being hardwired into the language.
The type of functions from \N\ to \N\ is represented by a codata type with one destructor \texttt{Ap}.
This destructor corresponds to the only way a function can be used, namely to apply it to an argument.
In the symmetric setting \texttt{Ap} takes two arguments, the producer argument $x$ for the value passed to the function, and the consumer argument $k$ for the consumer to be used on the result of the evaluation of the function\footnote{Readers familiar with linear logic might recognize this as the following decomposition of the function type: $\phi \multimap \psi = (\phi \otimes \psi^\bot)^\bot$.}.
\newcommand{\nattonat}{\N\hspace{-0.08cm}\shortrightarrow\hspace{-0.08cm}\N}
{\footnotesize
  \begin{align*}
    \codata\ \mathbf{type}\ \nattonat\ \{\ \texttt{Ap}(x \prd \N, k \con \N)\ \}
  \end{align*}
}%
The identity function $\lambda x.x$ can be written as a comatch.
{\footnotesize
  \begin{align*}
    \text{id} \coloneq \mathbf{match}_\codata\ \nattonat\ \{\ \texttt{Ap}(x,k) \Rightarrow x \gg k\ \} \prd \nattonat
  \end{align*}
}%
Codata types also allow for easy and intuitive programming with infinite structures.
For example, the type of streams of natural numbers is defined by the two destructors which give the head and the tail of a stream.
Note that the \texttt{Head} destructor does not directly return the first element; instead, a continuation for \N\ has to be passed as an argument.
\newcommand{\natstream}{\ensuremath{\N\hspace{-0.08cm}-\hspace{-0.08cm}\text{Stream}}}
{\footnotesize
  \begin{align*}
    &\codata\ \mathbf{type}\ \natstream\ \{\ \texttt{Head}(k \con \N)\ ;\ \texttt{Tail}(k \con \natstream)\ \}
   \end{align*}
}%
Codata also formalizes one essential aspect of object-oriented programming: programming against an interface \citep{cook90object, cook09understanding}.
We give a simple example of a customer ``interface'' with name and address fields; an ``object'' that implements the interface
could again be constructed with a copattern match.
{\footnotesize
  \begin{align*}
    &\codata\ \mathbf{type}\ \text{Customer}\ \{\ \texttt{Name}(k \con \text{String})\ ;\ \texttt{Address}(k \con \text{Address})\ \}
  \end{align*}
}%
In the next section we will see how symmetric data and codata types can be transformed back and forth other using the \derefunc\ algorithms, and why
we only need a single algorithm in the novel symmetric setting.

\subsection{Defunctionalization and Refunctionalization}
Having symmetric data and codata types turns \derefunc\ into \emph{one} algorithm, like we promised in the introduction.
We illustrate this with the example program in \cref{fig:natNumbers}.
In \cref{fig:natNumbers:data} we present \N\ as a data type, in \cref{fig:natNumbers:codata} as a codata type and in \cref{fig:natNumbers:matrix} as a matrix.
\Derefunc\ is now just matrix transposition; the terms themselves remain unchanged. In the asymmetric setting of previous work \cite{uroboro2015,ostermann2018dualizing,Uroboro2019}, complications due to the asymmetry prevented this easy formulation. For instance, asymmetric destructors are defunctionalized to functions with a ``special'' argument (the \texttt{this} object, in OO terminology), so different
kinds of function declarations and function calls had to be distinguished in the formalization.
\begin{figure}[bp]
  \begin{subfigure}{0.47\linewidth}
    {\footnotesize
      \begin{flalign*}
        \quad&\data\ \N\ \mathbf{where}&\\[-4pt]
        &\quad \texttt{Zero} \\[-4pt]
        &\quad \texttt{Suc}(x \prd \N) \\
        \intertext{\emph{with consumers}}
        &\texttt{pred}(k \con \N) \coloneq \match{\data}\ \N \\[-4pt]
        &\quad \texttt{Zero} \Rightarrow \texttt{Zero} \gg k \\[-4pt]
        &\quad \texttt{Suc}(x) \Rightarrow x \gg k \\
        &\texttt{add}(y \prd \N, k \con \N) \coloneq \match{\data}\ \N \\[-4pt]
        &\quad \texttt{Zero} \Rightarrow  y \gg k \\[-4pt]
        &\quad \texttt{Suc}(x) \Rightarrow x \!\gg \texttt{add}(\texttt{Suc}(y), k)
      \end{flalign*}
    }%
    \caption{Type $\N$ in data form.}
    \Description{Code snippet of a program where the type of natural numbers is in data form.}
    \label{fig:natNumbers:data}
  \end{subfigure}
  \hspace{-.3em}
  \begin{subfigure}{0.49\linewidth}
    {\footnotesize
      \begin{flalign*}
        \quad&\codata\ \N\ \mathbf{where}&\\[-4pt]
        &\quad \texttt{pred}(k \con \N) \\[-4pt]
        &\quad \texttt{add}(y \prd \N, k \con \N) \\
        \intertext{\emph{with producers}}
        &\texttt{Zero} \coloneq \match{\codata}\ \N \\[-4pt]
        &\quad \texttt{pred}(k) \Rightarrow \texttt{Zero} \gg k \\[-4pt]
        &\quad \texttt{add}(y,k) \Rightarrow y \gg k \\
        &\texttt{Suc}(x \prd \N) \coloneq \match{\codata}\ \N \\[-4pt]
        &\quad \texttt{pred}(k) \Rightarrow  x \gg k \\[-4pt]
        &\quad \texttt{add}(y,k) \Rightarrow x \!\gg \texttt{add}(\texttt{Suc}(y), k)
      \end{flalign*}
    }%
    \caption{Type $\N$ in codata form.}
    \Description{Code snippet of a program where the type of natural numbers is in codata form.}
    \label{fig:natNumbers:codata}
  \end{subfigure}

  \noindent
  \begin{subfigure}{0.49\textwidth}
    {\footnotesize
    \[
    \begin{array}{c|cc}
      \N& \texttt{Zero} & \texttt{Suc}(x \prd: \N) \\
      \hline
      \texttt{pred}(k \con \N) & \texttt{Zero} \gg k  & x \gg k \\
      \texttt{add}(y \prd \N, k \con \N) & y \gg k & x \gg \texttt{add}(\texttt{Suc}(y), k)
    \end{array}
    \]
    }
    \caption{Type $\N$ in matrix form.}
    \Description{The type of natural numbers presented as a matrix.}
    \label{fig:natNumbers:matrix}
  \end{subfigure}
  \caption{Type $\N$ in data, codata and matrix form.}
  \label{fig:natNumbers}
\end{figure}

For the system presented so far, and for the constructs used in \cref{fig:natNumbers}, it is clear after some reflection that \derefunc\ is semantics-preserving.
This is essentially due to the fact that the program still contains the same ``rewrites'' (cases of (co)pattern matches) - only the place where the rewrites are defined changes -, and there is no notion of evaluation order; all control flow is completely explicit.
In the next subsection we introduce additional control flow constructs which introduce the need for evaluation strategies.

\subsection{Evaluation Order}
\label{subsec:mainideas:evaluationorder}

In order to speak about evaluation order, we need to add additional constructs which let us form new commands for which the evaluator has to make a \emph{choice} on what to evaluate next.
We do this by introducing the $\mu$ and $\tilde\mu$ abstractions from the \lambdamutildemu\ calculus of \citet{curien00duality}.
These two constructs introduce new ways to form producers and consumers.

A $\mu$-abstraction $\muprdabs{x}{T}{c}$ is a producer of $T$ which abstracts over a consumer $x$ of $T$ in the command $c$.
Dually, a $\tilde\mu$-abstraction $\muconabs{x}{T}{c}$ is a consumer of $T$ which abstracts over a producer $x$ of $T$.
Note that in distinction to \cite{curien00duality} we use different type annotations on the variable instead of the tilde to distinguish the two kind of abstractions.
This allows us to generalize over them syntactically in the formalization.

We illustrate their intuitive meaning using the data type of natural numbers and the functions given in \cref{fig:natNumbers:data}.
The only terms that can be typed as producers of \N\ in the system without $\mu$ abstractions are generated by the grammar $v \coloneq \text{Zero}\, |\, \text{Suc}(v)\, |\, x$.
In particular, there was no way to type an analogue of $2+2$ as a producer of \N\ with is not yet fully evaluated.
Using a $\mu$-abstraction, we can now form the producer $\muprdabs{k}{\N}{(2 \gg \text{add}(2,k))}$.
The $\tilde\mu$ abstraction for \N, on the other hand, behaves more like a let-binding for producers in a command.
For example, the command $5 \gg \muconabs{x}{\N}{c}$ behaves like $\mathbf{let}\ x \prd \N \coloneq 5\ \mathbf{in}\ c$.

Once we have both $\mu$ and $\tilde\mu$ abstractions in the language confluence is lost, as witnessed by the following critical pair.
In this example, $\Omega$ refers to some unspecified non-terminating command.
Depending on which abstraction we evaluate first, we obtain either $\Omega$ or $\Done$.
This is clearly unsatisfactory.
\begin{equation}
  \begin{tikzcd}[column sep=huge]
    \Omega & \ar[l,"\text{\parbox{1.5cm}{\centering substitute $x$\\[-.8em]in producer}}"', "\reducestoinv"]
              \muprdabs{x}{T}{\Omega} \gg \muconabs{y}{T}{\Done}
              \ar[r,"\text{\parbox{1.5cm}{\centering substitute $y$\\[-.8em]in consumer}}", "\reducesto"']
           & \Done
  \end{tikzcd}
\label{eq:critical-pair}
\end{equation}
We regain confluence by prescribing an evaluation strategy for the redex above; either call-by-value or call-by-name.
In call-by-value, only producers which belong to the more restrictive grammar, not containing $\mu$-abstractions, can be substituted for producer variables.
In particular, we cannot substitute the producer corresponding to $2+2$ for a producer variable of type \N.
More generally:
Under call-by-value we evaluate $\mu$-abstractions before $\tilde\mu$ abstractions, and conversely for call-by-name.

This leads to the problem of which evaluation order to choose for the system.
This choice should strive to maximize the number of valid $\eta$-equalities.
For example, the $\eta$ equality for the function type, assuming that $x$ and $k$ do not occur free in $e$, is:
\begin{align*}
   \match{\codata}\ \nattonat\ \{\ \texttt{Ap}(x,k) \Rightarrow e \gg \texttt{Ap}(x,k)\ \} \etaeq e
\end{align*}
This rule makes the following two commands $\eta$-equivalent:
\begin{gather*}
  \match{\codata}\ \nattonat\ \{ \texttt{Ap}(x,k) \Rightarrow \muprdabs{y}{\nattonat}{\Omega} \gg \texttt{Ap}(x,k) \} \gg \muconabs{x}{\nattonat}{\Done} \\
  \muprdabs{y}{\nattonat}{\Omega} \gg \muconabs{x}{\nattonat}{\Done}
\end{gather*}
But under the call-by-value strategy they evaluate to different results.
In summary, we can thus observe that generally, $\eta$-rules are only valid for data types when evaluated under call-by-value, while they are only valid for codata types when evaluated under call-by-name.
We can turn this observation into an evaluation strategy:
Under the \emph{polar} evaluation order, all data types are evaluated using call-by-value and all codata types using call-by-name.

\subsection{\Derefunc\ and evaluation order}
\label{subsec:mainideas:derefunceval}
What consequences does the introduction of evaluation order have for \derefunc?
The central observation is that we can combine the simple \derefunc\ algorithm with global \cbv\ or \cbn, \textit{but we cannot combine this simple algorithm with the polar evaluation order.}
In order to see this, consider the critical pair in \cref{eq:critical-pair}.
Whatever evaluation order we choose for this redex, we must use the same evaluation order for it after defunctionalizing or refunctionalizing the type $T$.
This is guaranteed if we fix a global evaluation order, but it fails if the evaluation order is dependent on whether $T$ is a data or codata type.

The following table summarizes this situation.
If we choose call-by-value we loose the $\eta$-equalities for codata types, and if we choose call-by-name we loose the $\eta$-equalities for data types.
But in both cases we can use the simple \derefunc\ algorithm described above.
If we choose the polar evaluation order, on the other hand, all $\eta$-equalities are valid, but the \derefunc\ algorithm is \emph{no longer semantics-preserving}.

\begin{center}
  \begin{tabular}{cccc}
    \toprule
    Eval Order & $\eta$ for data & $\eta$ for codata & De/Refunc. \\
    \midrule
    Global \cbv & \cmark & \xmark & \cmark \\
    Global \cbn & \xmark & \cmark & \cmark \\
    Polar & \cmark & \cmark & \xmark \\
    \bottomrule
  \end{tabular}
\end{center}

How do we combine \derefunc\ with the polar evaluation order, which validates all $\eta$-equalities?
To solve this problem, we introduce a fourth evaluation order: the ``nominal'' evaluation order, where each type explicitly declares whether it should be evaluated by-value or by-name.
This scheme allows the declaration of call-by-name data types and call-by-value codata types.
However, we consider these two new types to be mere intermediate steps in the translation from call-by-value data types to call-by-name codata types, and vice versa, since they are not well-behaved when we consider their $\eta$-laws.

\subsection{Summary}%
\label{subsec:mainideas:summary}

\Cref{fig:commutativediagram} gives a concise summary of this paper.
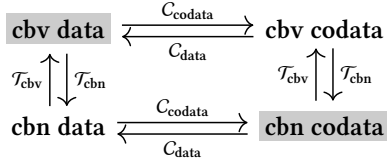
\begin{figure}[h]
  \centering
  \begin{tikzcd}
    \greybox{\cbv\ \data} \ar[rr, shift left=0.5ex, "\refunctionalize"] \ar[d, shift left=0.75ex, "\cbvtocbn"] & &
    \cbv\ \codata \ar[ll, shift left=0.5ex, "\defunctionalize"] \ar[d, shift left=0.75ex, "\cbvtocbn"] \\
    \cbn\ \data \ar[rr, shift left=0.5ex, "\refunctionalize"] \ar[u, shift left=0.75ex, "\cbntocbv"] & &
    \greybox{\cbn\ \codata} \ar[ll, shift left=0.5ex, "\defunctionalize"] \ar[u, shift left=0.75ex, "\cbntocbv"]
  \end{tikzcd}
  \caption{All possible transformations of a type.}
\label{fig:commutativediagram}
\end{figure}

We factor the \emph{full \derefunc} transformation $\fxfunctionalize$ into two simpler transformations: \emph{core defunctionalization}\footnote{In the following, we will often refer to core \derefunc\ simply as just \derefunc.} $\defunctionalize$  and \emph{core refunctionalization} $\refunctionalize$ exchange data and codata, as well as \emph{\cbxtocby}, $\cbvtocbn$ and $\cbntocbv$.
The former will perform the essential part of \derefunc, i.e. exchanging data with codata types and vice versa without any changes to the evaluation strategy.
Correspondingly, the latter will only change the evaluation strategy, while keeping the semantics of all existing expressions intact.
For this transformation we have to add appropriate \emph{shift types} $\cbnshift{T}$ and $\cbvshift{T}$ to the program for every type $T$.
Shift types were originally introduced by \citet{girard2001} in the context of ludics and polarized logic, where they mediate between positive and negative types.
Their use to control the evaluation order has been described more accessibly by \citet{zeilberger2008unity} and \citet{beyondPolarity}.
These transformations make the diagram in \cref{fig:commutativediagram} commute up to the insertion of some double-shifts, which can be removed in the special case of a mere round-trip.
This results in the full defunctionalization $\fdefunctionalize = \cbntocbv \circ \defunctionalize$ and full refunctionalization $\frefunctionalize = \cbvtocbn \circ \refunctionalize$.

An example of a program which is transformed along the path \emph{top left} $\to$ \emph{top right} $\to$ \emph{bottom right} in \cref{fig:commutativediagram} can be found in \cref{fig:exampletrans}.
Note that all \texttt{main} commands evaluate to \Done.

\begin{figure}[tbp]
  \begin{subfigure}{.9\linewidth}
    {\footnotesize
      \begin{align*}
        & \cbv\ \data\ \mathbf{type}\ \N\ \{\ \texttt{Zero}\ ;\ \texttt{Suc}(x \prd \N)\ \} \\
        & \texttt{pred}(k \con \N) \coloneq \match{\data}\ \N \\[-.4em]
        & \quad \texttt{Zero} \Rightarrow \texttt{Zero} \gg k\\[-.4em]
        & \quad \texttt{Suc}(n) \Rightarrow n \gg k\\
        & \texttt{main} \coloneq \muprdabs{k}{\N}{\Done} \gg \muconabs{n}{\N}{\Omega}
      \end{align*}
    }%
    \subcaption{Original program.}
  \end{subfigure}

  \hfill
  \begin{subfigure}{.44\linewidth}
    {\footnotesize
      \begin{align*}
        & \cbv\ \codata\ \mathbf{type}\ \N\ \{\ \texttt{pred}(k \con \N)\ \} \\
        \\
        & \texttt{Zero} \coloneq \match{\codata}\ \N \\[-.4em]
        & \quad \texttt{pred}(k)\Rightarrow \texttt{Zero} \gg k\\[-.4em]
        & \texttt{Suc}(n \prd \N) \coloneq \match{\codata}\ \N \\[-.4em]
        & \quad \texttt{pred}(k) \Rightarrow n \gg k\\
        & \texttt{main} \coloneq \muprdabs{k}{\N}{\Done} \gg \muconabs{n}{\N}{\Omega}
      \end{align*}
    }%
    \subcaption{Intermediate step after refunctionalizing (a) by applying \refunctionalize.}
  \end{subfigure}
  \hfill
  \begin{subfigure}{.45\linewidth}
    {\footnotesize
      \begin{align*}
        & \cbn\ \codata\ \mathbf{type}\ \N\ \{\ \texttt{pred}(k \con \cbvshift{\N})\ \} \\[-.4em]
        & \cbv\ \data\ \mathbf{type}\ \cbvshift{\N}\ \{\ \cbvterm{x \prd \N}\ \} \\
        & \texttt{Zero} \coloneq \match{\codata}\ \N \\[-.4em]
        & \quad \texttt{pred}(k)\Rightarrow \cbvterm{\texttt{Zero}} \gg k\\[-.4em]
        & \texttt{Suc}(n \prd \cbvshift{\N}) \coloneq \match{\codata}\ \N \\[-.4em]
        & \quad \texttt{pred}(k) \Rightarrow n \gg k\\
        & \texttt{main} \coloneq \muprdabs{k}{\cbvshift\N}{\Done} \gg \muconabs{n}{\cbvshift\N}{\Omega}
      \end{align*}
    }%
    \subcaption{Final result after performing cbv-to-cbn translation \cbvtocbn\ of (b).}
  \end{subfigure}%
  \hspace*{\fill}
  \caption{A simple program which is transformed without changing semantics using $\frefunctionalize = \cbvtocbn \circ \refunctionalize$.}
\label{fig:exampletrans}
\end{figure}

\section{Formalization}
\label{sec:formalization}
In this section we will formally describe the syntax, type system and operational semantics of the CPS fragment of our language.
We will extend this language in \cref{sec:evaluationorder} with constructs which allow to write programs in direct style.
We use the notation $\overline{X}$ to represent a (possibly empty) sequence $X_1,\ldots,X_i,\ldots,X_n$.
We follow Featherweight Java~\citep{igarashi2001featherweight} conventions for this notation.
In this convention, multiple occurrences of such sequences should be read as being indexed simultaneously, e.g. the judgement ``$\Gamma, \overline{\Pi} \vdash \overline{c} \cmd{}$'' should be read as a list of judgements ``$\Gamma, \Pi_i \vdash c_i \cmd{}$''.

The inherent symmetry of the system allows us to minimize the number of rules, since we don't have to write down separate rules for data and codata types, matches and comatches, constructors and destructors.
In order to do this, we introduce polarities $p$ which can either be $\data$ or $\codata$.
For example, a $\match{\data}$ is a pattern match and a $\match{\codata}$ is a copattern match and so on.
The syntax of our language is defined in \cref{fig:formalization:terms}, and the typing and well-formedness rules are given in \cref{fig:formalization:types,fig:formalization:well-formed}, which we will now explain in turn.

We use the set \textsc{Tname} for names of types, and the set \textsc{Name} for for names of constructors, destructors and functions.
Whenever something can stand for either a constructor or destructor, we refer to t as an \emph{xtor}.
All the typing rules implicitly have a program in their context, and we use lookup functions to obtain various information about the declarations in that program.

\begin{figure}[htbp]
  \begin{minipage}{\linewidth}
    \[
      \begin{array}{lclr}
        \multicolumn{4}{r}{x,y \in \textsc{Var} \quad T \in \textsc{Tname} \quad \mathcal{X}, \mathcal{Y} \in \textsc{Name} \mspace{100mu}   \emph{Variables and Names}} \\[0.1cm]

        p & \Coloneqq & \data \sep \codata & \emph{Polarity} \\
        s & \Coloneqq & \cbv \sep \cbn & \emph{Evaluation Strategy} \\
        o & \Coloneqq & \mathbf{prd} \sep \mathbf{con} & \emph{Orientation} \\[0.2cm]
        \Gamma, \Delta, \Pi & \Coloneqq & \diamond \sep \Gamma, x \prdcon T  & \emph{Contexts} \\
        \sigma,\tau & \Coloneqq & () \sep (\sigma, e)  & \emph{Substitutions} \\[0.2cm]

        e & \Coloneqq & x \sep \mathcal{X}\sigma \sep \match{p}\, T\, \lbrace\, \overline{\mathcal{X}\Delta \Rightarrow c}\, \rbrace & \emph{Expressions} \\
        c & \Coloneqq &  e \gg e  \sep \Done & \emph{Commands}  \\[0.2cm]

        d & \Coloneqq & s\ p\ \mathbf{type}\ T\ \lbrace\ \overline{\mathcal{X}\Delta}\ \rbrace\ \mathbf{with}\ \overline{f} & \emph{Type declarations}  \\
        f & \Coloneqq & \mathcal{X}\Pi \coloneq \match{p}\ T\ \{ \overline{\mathcal{Y}\Delta \Rightarrow c} \} & \emph{Functions} \\
        P & \Coloneqq & (\overline{d},\ c) & \emph{Program} \\[0.2cm]
      \end{array}
    \]
  \end{minipage}
  \small
  \begin{align*}
    \valfun{\data} &\coloneq \mathbf{prd}
      & \cntfun{\data} &\coloneq \mathbf{con}
      & \widehat{\data} &\coloneq \codata 
      & \widehat{\cbv} &\coloneq \cbn 
      &\widehat{\mathbf{prd}} &\coloneq \mathbf{con}\\[-.3em]
    \valfun{\codata} &\coloneq \mathbf{con}
      & \cntfun{\codata} &\coloneq \mathbf{prd}
      & \widehat{\codata} &\coloneq \data
      & \widehat{\cbn} &\coloneq \cbv
      & \widehat{\mathbf{con}} &\coloneq \mathbf{prd}
  \end{align*}%
  \caption{Program declarations and terms, and some helper functions.}
  \Description{Program declarations and terms.}
  \label{fig:formalization:terms}
\end{figure}

\begin{figure}[p!]
  \begin{flushright}
    \fbox{Context Formation: $\vdash \Gamma$}
  \end{flushright}
  \begin{minipage}{0.24\linewidth}
    \begin{prooftree}
      \AxiomC{\phantom{$\Gamma$}}
      \RightLabel{\textsc{T-Ctx}$_1$}
      \UnaryInfC{$\vdash \diamond$}
    \end{prooftree}
  \end{minipage}
  \begin{minipage}{0.64\linewidth}
    \begin{prooftree}
      \AxiomC{$\vdash \Gamma \quad T \in \text{Prog} \quad x \not\in \text{dom}(\Gamma)$}
      \RightLabel{\textsc{T-Ctx}$_2$}
      \UnaryInfC{$\vdash \Gamma, x \prdcon T$}
    \end{prooftree}
  \end{minipage}
  \begin{flushright}
    \fbox{Substitution typing: $\Gamma \vdash \sigma : \Delta$}
  \end{flushright}
  \begin{minipage}{0.4\textwidth}
    \begin{prooftree}
      \AxiomC{\phantom{$\Gamma$}}
      \RightLabel{\textsc{T-Subst}$_1$}
      \UnaryInfC{$\Gamma \vdash () : \diamond$}
    \end{prooftree}
  \end{minipage}
  \begin{minipage}{0.5\textwidth}
    \begin{prooftree}
      \AxiomC{$\Gamma \vdash \sigma : \Delta$}
      \AxiomC{$\Gamma \vdash e \prdcon T$}
      \AxiomC{$\substitutable(e)$}
      \RightLabel{\textsc{T-Subst}$_2$}
      \TrinaryInfC{$\Gamma \vdash (\sigma, e) : \Delta, x \prdcon T$}
    \end{prooftree}
  \end{minipage}
  \begin{flushright}
    \fbox{Expression typing: $\Gamma \vdash e \prdcon T$}
  \end{flushright}
  \begin{prooftree}
    \AxiomC{$\Gamma(x) = (o,T)$}
    \RightLabel{\textsc{T-Var}}
    \UnaryInfC{$\Gamma \vdash x \prdcon T$}
  \end{prooftree}

 \begin{prooftree}
    \AxiomC{$\mathcal{X}\Delta \in \text{Xtors}(T)$}
    \AxiomC{$\text{Polarity}(T) = p$}
    \AxiomC{$\Gamma \vdash \tau : \Delta$}
    \RightLabel{\textsc{T-Xtor}}
    \TrinaryInfC{$\Gamma \vdash \mathcal{X}\tau \overset{\valfun{p}}{:} T$}
  \end{prooftree}

  \begin{prooftree}
    \AxiomC{$\mathcal{X}\Pi \coloneq \ldots \in \text{Funs}(T)$}
    \AxiomC{$\text{Polarity}(T) = p$}
    \AxiomC{$\Gamma \vdash \tau : \Pi$}
    \RightLabel{\textsc{T-Fun}}
    \TrinaryInfC{$\Gamma \vdash \mathcal{X}\tau \overset{\cntfun{p}}{:} T$}
  \end{prooftree}

  \begin{prooftree}
    \AxiomC{$\Gamma, \overline{\Delta} \vdash \overline{c} \cmd{}$}
    \AxiomC{(Side condition: see text)}
    \RightLabel{\textsc{T-Match}}
    \BinaryInfC{$\Gamma \vdash \match{p}\, T\ \lbrace\,  \overline{\mathcal{X}\Delta \Rightarrow c}\, \rbrace \overset{\cntfun{p}}{:} T$}
  \end{prooftree}

  \begin{flushright}
    \fbox{Command typing: $\Gamma \vdash c \cmd$}
  \end{flushright}
  \begin{minipage}{0.59\linewidth}
    \begin{prooftree}
      \AxiomC{$\Gamma \vdash e_1 \prd T$}
      \AxiomC{$\Gamma \vdash e_2 \con T$}
      \RightLabel{\textsc{T-Cut}}
      \BinaryInfC{$\Gamma \vdash e_1 \gg e_2 \cmd$}
    \end{prooftree}
  \end{minipage}
  \begin{minipage}{0.39\linewidth}
    \begin{prooftree}
      \AxiomC{\phantom{$\prd$}}
      \RightLabel{\textsc{T-Done}}
      \UnaryInfC{$\Gamma \vdash \Done \cmd$}
    \end{prooftree}
  \end{minipage}
  \caption{Typing rules.}
  \label{fig:formalization:types}
\end{figure}

\begin{figure}[p!]
  \begin{prooftree}
    \AxiomC{$\Pi \vdash \match{p}\ T\ \{ \overline{\mathcal{X}\Delta \Rightarrow c} \} \overset{\cntfun{p}}{:} T$}
    \RightLabel{\textsc{Wf-Fun}}
    \UnaryInfC{$\vdash \mathcal{X}\Pi \coloneq \match{p}\ T\ \{ \overline{\mathcal{X}\Delta \Rightarrow c} \} \overset{\cntfun{p}}{:} T\ \OK$}
  \end{prooftree}

  \begin{minipage}{0.5\textwidth}
    \begin{prooftree}
      \AxiomC{$\vdash \overline{\Gamma}$}
      \AxiomC{$\vdash \overline{f}\ \OK$}
      \RightLabel{\textsc{Wf-Type}}
      \BinaryInfC{$\vdash s\ p\ \mathbf{type}\ T\ \lbrace\ \overline{\mathcal{X}\Delta}\ \rbrace\ \mathbf{with}\ \overline{f}\ \OK $}
    \end{prooftree}
  \end{minipage}
  \begin{minipage}{0.4\textwidth}
    \begin{prooftree}
      \AxiomC{$\vdash \overline{d}\ \OK$}
      \AxiomC{$\vdash c \cmd$}
      \RightLabel{\textsc{Wf-Prog}}
      \BinaryInfC{$\vdash (\overline{d}, c)\ \OK$}
    \end{prooftree}
  \end{minipage}
  \caption{Well-formedness rules.}
  \label{fig:formalization:well-formed}
\end{figure}

\begin{figure}[p!]
  \[
    \begin{array}{ccc}
    & \mathcal{X}\sigma \mkCmd \match{\data}\, T\, \lbrace\, \mathcal{X}\Delta \Rightarrow c;\, \ldots\,\rbrace\ \reducesto c\, \sigma & \textsc{Match} \\
    & \match{\codata}\, T\, \lbrace\, \mathcal{X}\Delta \Rightarrow c;\, \ldots\,\rbrace\ \mkCmd \mathcal{X}\sigma \reducesto c\, \sigma & \textsc{Comatch} \\
    (\mathcal{X}\Pi \coloneq \mathcal{Y}\Delta \Rightarrow c;\ldots) \in \textsc{Prog}
      & \mathcal{Y}\tau \mkCmd \mathcal{X}\sigma\ \reducesto\ (c\ \tau)\,\sigma & \textsc{ConCall} \\
    (\mathcal{X}\Pi \coloneq \mathcal{Y}\Delta \Rightarrow c;\ldots) \in \textsc{Prog}
      & \mathcal{X}\sigma \mkCmd \mathcal{Y}\tau \reducesto\ (c\ \tau)\,\sigma & \textsc{PrdCall}
    \end{array}
  \]
  \caption{Operational Semantics.}
  \label{fig:formalization:semantics}
\end{figure}

\subsection{Type declarations and the program}
\label{subsubsec:formalization:program}

A \emph{program} $P$ consists of a list of data and codata type declarations $d$ and an entry point in the form of a top-level command.
In order to check that a program is well-formed, we have to use the rule \textsc{Wf-Prog} to check that the entry point typechecks as a command, which in turn uses the rule \textsc{Wf-Type} to check that each of the type declarations is correct.
Note that we do not consider declarations to be ordered and all data and codata types, constructors, destructors and functions may reference one another in mutual recursion.
We also require all names that are used to be unambiguous, but we don't write down the obvious rules.

A \emph{type} declaration $s\ p\ \mathbf{type}\ T\ \lbrace\ \overline{\mathcal{X}\Delta}\ \rbrace\ \mathbf{with}\ \overline{f}$ introduces a data or codata type $T$ (with evaluation strategy $s$) by specifying both its xtors $\mathcal{X}_i\Delta_i$ and a list of functions $f_i$ which pattern match on its xtors.
The only reason why the xtors of a type have to be declared together with the functions matching on them is to allow for a simpler presentation of the algorithm in \cref{sec:xtorization}; a real programming language implementing these ideas would not use this restriction.
Each function declaration $\mathcal{Y}\Pi \coloneq \match{p}\,T\, \{ \overline{\mathcal{X}\Delta} \Rightarrow \overline{c} \}$ declares a function $\mathcal{X}$ with arguments $\Pi$ by (co)pattern matching on all xtors of the type $T$ to which they belong.
The rule \textsc{Wf-Fun} uses the expression typing judgement introduced above to typecheck these global (co)pattern matches.

\subsection{Contexts and Substitutions}
\label{subsec:formalization:contexts}

The definition of data and codata types, constructors and destructors, pattern matching and copattern matching can be formulated more concisely and uniformly by using \emph{contexts} and \emph{substitutions}.
Consider the example program from \cref{fig:natNumbers:data}, where natural numbers are defined by two constructors: \texttt{Zero} and \texttt{Suc}.
\texttt{Zero} does not bind anything, but \texttt{Suc} is defined using the context $\Delta = x \prd \mathbb{N}$.
This context $\Delta$ determines simultaneously that pattern matching on a \texttt{Suc} constructor extends the context by $\Delta$, and that in order to construct a natural number with \texttt{Suc}, a substitution $\sigma$ for $\Delta$ has to be provided.

Contexts map variables from the set \textsc{Var} to their type and orientation (producer or consumer) and are constructed according to the rules \textsc{T-Ctx}$_1$ and \textsc{T-Ctx}$_2$ from \cref{fig:formalization:types}.
A substitution $\sigma$ for a context $\Delta$ consists of one expression for each variable in $\Delta$.
We use the typing judgement $\Gamma \vdash \sigma : \Delta$ to express that each expression in $\sigma$ can be typed in the context $\Gamma$ with the type and orientation specified in the context $\Delta$.
The meaning of the $\substitutable(e)$ premise in the rule \textsc{T-Subst}$_2$ will be explained in \cref{sec:evaluationorder} and is irrelevant for now, as it holds for all expressions in the current setting.
Note that there is an obvious identity substitution \idmorphism{\Gamma} for every context which satisfies $\Pi, \Gamma \vdash \idmorphism{\Gamma} : \Gamma$.

\subsection{Expressions and Commands}
\label{subsubsec:formalization:expressions}

While we have only one syntactic category $e$ of expressions, there are two different typing judgements for expressions; expressions can either be typed as a producer with $\Gamma \vdash e \prd T$ or as a consumer with $\Gamma \vdash e \con T$.
In all, there are four different rules which govern the typing of expressions.

The rule \textsc{T-Var} is quite self-explanatory, we just have to look up both the type and the orientation of the variable in the context.

The rule \textsc{T-Xtor} covers the typing of both constructors and destructors, which we collectively call \emph{xtors}.
Constructors have to be typed as producers, and destructors as consumers; the rule accomplishes this by looking up the type $T$ to which the xtor belongs in the program, and the polarity $p$ of that type.
Looking up the xtor in the program also tells us what arguments we have to provide in the substitution $\tau$, and the helper function $\valfun{p}$ guarantees that the result is typechecked with the correct orientation.

The rule \textsc{T-Fun} is very similar to the rule \textsc{T-Xtor}.
Instead of typechecking constructors and destructors, it governs the call of functions declared in the program.
The difference between \textsc{T-Xtor} and \textsc{T-Fun} is that we don't look up the signature of a constructor or destructor, but the declaration of a function, i.e.\ its signature together with its body.
If the polarity of the type to which the function belongs is \data, then the function is defined by pattern matching and the function call has to be typed as a consumer; similarly, if the polarity is \codata, the function is defined by copattern matching and the call must be typed as a producer.
The helper function $\cntfun{p}$ ensures that this is the case.

The rule \textsc{T-Match} covers local pattern and copattern matches.
In that rule we have to check that the right-hand side of each case typechecks as a command.
In each case we extend the outer context with the context bound by the constructor or destructor.
Pattern matching has to be exhaustive, and the arguments bound in the case have to be identical to the ones declared in the program.
We have omitted these requirements in the formulation of the rule \textsc{T-Match} in order to keep it more legible.

There are two ways to form commands.
A logical command $e_1 \gg e_2$ consists of two expressions; the expression $e_1$ has to typecheck as producer and the expression $e_2$ as a consumer of the same type.
The rule is named after the ``Cut'' rule from sequent calculus, to which it corresponds.
The command \Done\ typechecks in any context and corresponds to a successfully terminated computation.

\subsection{Operational Semantics}
\label{subsubsec:formalization:semantics}

Reduction only applies to closed commands, and the rules are given in \cref{fig:formalization:semantics}.
Suppose that $c$ is well-typed in the context $\Gamma,\Delta$, and that $\sigma$ is a substitution from the empty context for $\Delta$, i.e. $\diamond \vdash \sigma : \Delta$.
Then the result of substituting $\sigma$ in $c$ for the variables from $\Delta$, which we write $c\ \sigma$, is well-defined and typechecks in the context $\Gamma$.
We don't give the full rules for substitution, since these are obvious but involve the usual technical complications of variable capture.
Since all the xtors of the corresponding type must occur exactly once inside a match, their order does not matter and we will thus adopt the notational convention that the matching case in a pattern match is written as its first case.
When a value (constructor or destructor) meets a continuation (pattern match or copattern match), evaluation proceeds by straightforward substitution.
The case is slightly more difficult if a constructor meets the call of a globally defined function.
In that case we have the two contexts, $\Gamma$ and $\Delta$, where the function $\mathcal{X}$ is declared in the program.
In that case $c$ is typed in the context $\Gamma, \Delta$, and $\sigma$ and $\tau$ are substitutions for $\Gamma$ and $\Delta$, respectively.
The \Done\ command is in normal form and cannot be evaluated further.

\subsection{Type Soundness}
\label{subsubsec:formalization:soundness}

The soundness of the system has been mechanically verified in the accompanying Coq formalization.
The following two theorems hold.

\begin{theorem}[Preservation]
  If $\diamond \vdash c_1 \cmd$ and $c_1 \reducesto c_2$, then $\diamond \vdash c_2 \cmd$
\label{thm:preservation}
\end{theorem}

\begin{theorem}[Progress]
  If $\diamond \vdash c_1 \cmd$ then either $c_1 = \Done$ or there exists a command $c_2$ such that $c_1 \reducesto c_2$.
\label{thm:progress}
\end{theorem}

\section{Symmetric \Derefunc}
\label{sec:xtorization}
In this section we explain how to change the polarity of a type $T$ via core defunctionalization \defunctionalize\ and core refunctionalization \refunctionalize, which are subsumed by one algorithm \xfunctionalize.
That is, we show how to implement the horizontal transformations of the diagram introduced in \cref{sec:mainideas}.

\begin{center}
  \begin{tikzcd}
    \cbv\ \data \ar[rr, shift left=0.5ex, "\refunctionalize"] \ar[d, shift left=0.75ex, gray, draw=gray, "\cbvtocbn"] & &
    \cbv\ \codata \ar[ll, shift left=0.5ex, "\defunctionalize"] \ar[d, shift left=0.75ex, gray, draw=gray, "\cbvtocbn"] \\
    \cbn\ \data \ar[rr, shift left=0.5ex, "\refunctionalize"] \ar[u, shift left=0.75ex, gray, draw=gray, "\cbntocbv"] & &
    \cbn\ \codata \ar[ll, shift left=0.5ex, "\defunctionalize"] \ar[u, shift left=0.75ex, gray, draw=gray, "\cbntocbv"]
  \end{tikzcd}
\end{center}

We make one simplifying assumption; we assume that the program does not contain any local pattern matches on $T$.
This does not reduce the expressiveness of the system, but simplifies the presentation of the algorithm.%
\footnote{
  In order to lift this restriction, local pattern matches have to be first lambda-lifted and replaced by functions declarions.
  Only in the next step can the the algorithm presented in this section be used.
  Functions written by the programmer have to be distinguished from functions which are the result of lambda lifting by use of annotations.
  After \derefunc\ of the program, all functions which resulted from lambda-lifting have to be inlined.
  All details about annotations, lifting and inlining have been developed and formally verified by \citet{Uroboro2019}.
  }

In \cref{fig:formalization:terms} we used the same set \textsc{Name} for the names of both constructors/destructors and function calls.
This choice was motivated by the fact that core \derefunc\ should be the identity function on terms.
Typing derivations, on the other hand, are affected by \derefunc.
For example, if a constructor \X\ with arguments $\Delta$ is declared for the data type $T$, then a corresponding function declaration $\X$ will be declared for the codata type $T$ in the refunctionalized program.
That is, the the following type derivation for a term $\X \sigma$ on the top will be replaced by the derivation below, where $\mathcal{D}$ is some derivation for $\Gamma \vdash \sigma : \Delta$ and $\mathcal{D}'$ is the corresponding derivation for the same judgement within the refunctionalized program.

  \begin{prooftree}
    \AxiomC{$\mathcal{X}\Delta \in \textsc{Xtors}(T)$}
    \AxiomC{$\text{Polarity}(T)=\data$}
    \AxiomC{$\mathcal{D}$}
    \noLine
    \UnaryInfC{$\Gamma \vdash \sigma : \Delta$}
    \RightLabel{\textsc{T-Xtor}}
    \TrinaryInfC{$\Gamma \vdash \mathcal{X}\sigma \prd T$}
  \end{prooftree}
  \begin{prooftree}
    \AxiomC{$\mathcal{X}\Delta \in \textsc{Funs}(T)$}
    \AxiomC{$\text{Polarity}(T)=\codata$}
    \AxiomC{$\mathcal{D}'$}
    \noLine
    \UnaryInfC{$\Gamma \vdash \sigma : \Delta$}
    \RightLabel{\textsc{T-Fun}}
    \TrinaryInfC{$\Gamma \vdash \mathcal{X}\sigma \prd T$}
  \end{prooftree}
\vskip\baselineskip

\Derefunc\ $\xfunctionalize^T$ does not change any type declaration apart from the type declaration for $T$ itself, which we now define.
We do this by matrix transposition, as described by \citet{ostermann2018dualizing}.
A type declaration
\begin{equation*}
  s\ p\ \mathbf{type}\ T\ \lbrace\ \overline{\X\Delta} \rbrace\ \mathbf{with}\ \overline{\emph{f}}
\end{equation*}
is read into the first matrix of \cref{fig:xtorization:matrix} and then transposed to obtain the second matrix.
The second matrix is then used to generate the new type declaration.

\newcommand{\narrowcdots}{\hspace{-0.2cm}\small$\cdots$}
\begin{figure}[htbp]
  \scriptsize
  \begin{tikzpicture}
    \node (a) at (0,0) {
      \begin{tabular}{|c|ccc|}
        \hline
        \diagbox{\textsc{Dtor}\hspace{-1cm}}{\textsc{Fun}} & $\mathcal{X}_1\Gamma_1$ & \narrowcdots & \hspace{-0.2cm}$\mathcal{X}_m\Gamma_m$ \\
        \hline
        $\mathcal{Y}_1\Delta_1$ & $c_{1,1}$ & \narrowcdots & \hspace{-0.2cm}$c_{1,m}$ \\
        $\vdots$ & $\vdots$ &  & $\vdots$ \\
        $\mathcal{Y}_n\Delta_n$ & $c_{n,1}$ & \narrowcdots & \hspace{-0.2cm}$c_{n,m}$ \\
        \hline
      \end{tabular}
    };
    \node (b) at (6,0) {
      \begin{tabular}{|c|ccc|}
        \hline
        \diagbox{\textsc{Ctor}\hspace{-1cm}}{\textsc{Fun}} & $\mathcal{Y}_1\Delta_1$ & \narrowcdots & \hspace{-0.2cm}$\mathcal{Y}_n\Delta_n$ \\
        \hline
        $\mathcal{X}_1\Gamma_1$ & $c_{1,1}$ & \narrowcdots & \hspace{-0.2cm} $c_{n,1}$ \\
        $\vdots$ & $\vdots$ & & $\vdots$ \\
        $\mathcal{X}_m\Gamma_m$ & $c_{1,m}$ & \narrowcdots & \hspace{-0.2cm} $c_{n,m}$ \\
        \hline
      \end{tabular}
    };
    \draw[->,bend left,transform canvas={yshift=+.2cm}] (a.east) to node[above] {defunctionalize} (b.west) ;
    \draw[->,bend left,transform canvas={yshift=-.2cm}] (b.west) to node[below] {refunctionalize} (a.east) ;
  \end{tikzpicture}
  \caption{\Derefunc\ via matrix transposition.}
  \Description{\Derefunc\ via matrix transposition.}
  \label{fig:xtorization:matrix}
\end{figure}

For the system presented in \cref{sec:formalization}, the following properties are easily established.

\subsection{Properties of \Derefunc}
\label{subsec:xtorization:properties}

Since matrix transposition is its own inverse, it is obvious that defunctionalization and refunctionalization are mutually inverse:

\begin{theorem}[Mutual Inverse]
  \label{theorem:progsinverse}
  For every well-formed program $P$ and data type $T$ in $P$:
  \begin{align*}
    \xfunctionalizeinv^{T}(\xfunctionalize^{T}(P)) = P
  \end{align*}
\end{theorem}

\Derefunc\ also preserve well-formedness of programs.

\begin{theorem}[Typeability preservation]
  \label{theorem:typeabilitypreservation}
  If $P$ is a well-formed program, then $\xfunctionalize^{T}(P)$ is also well-typed.
\end{theorem}

\Derefunc\ is semantics-preserving.

\begin{theorem}[Semantic preservation]
  \label{theorem:semantic}
  Let $P$ be a program, with $\vdash P\ \OK$, $T$ a type of polarity $p$ in $P$ and $c_1$ a closed command in $P$, i.e. $\vdash c_1 \cmd$.
  Furthermore, $c_1$ may not contain any local pattern matches on $T$.
  Then, if $c_1 \reducesto c_2$ in $P$, $c_1 \reducesto c_2$ in $\xfunctionalize^T(P)$.
\end{theorem}
\begin{proof}
  Since this transformation only transposes the matrix, the cell which is addressed in this matrix by a pair of \textsc{Fun} and \textsc{Xtor} does not change, thus the body which will be substituted for a command consisting of such a pair will be the same as before.
\end{proof}

\section{Evaluation Order}
\label{sec:evaluationorder}
In this section we will formally introduce a new syntactic construct to the language: $\mu$-abstractions.
The introduction of $\mu$ abstractions will create \emph{critical pairs}.
A critical pair is a redex that can be evaluated to two different commands which don't reduce to the same normal form.
In order to ``defuse'' these critical pairs we have to introduce an evaluation order, which prescribes precisely how to evaluate these redexes.

After presenting several different alternative evaluation strategies we will end with a \emph{nominal} strategy.
Every type declares how it's redexes are to be evaluated.

\subsection{\texorpdfstring{Extending the calculus with $\mu$ abstractions}{Extending the calculus with mu abstractions}}
\label{subsec:evaluationorder:extension}

We will now extend the syntax and typing rules with $\mu$ abstractions in \cref{fig:muabstraction}.
In the system of \cref{sec:formalization} the syntax directly determined the only possible operational semantics.
But with the addition of $\mu$ the calculus is no longer confluent, since there is a critical pair which is well-known in the \lambdamutildemu\ literature:
A command where $\mu$ and $\tilde{\mu}$ meet.
In our setting, this corresponds to a producer $\muprdabs{x}{T}{c_1}$ and a consumer $\muconabs{x}{T}{c_2}$, as presented in \cref{subsec:mainideas:evaluationorder}.
\begin{equation*}
  c_1 \quad \reducestoinv \quad  \muprdabs{x}{T}{c_1} \gg \muconabs{y}{T}{c_2} \quad \reducesto \quad c_2
\end{equation*}
\begin{figure}[tbp]
  \[
    \begin{array}{lclr}
      e & \Coloneqq & \ldots \sep \muqabs{x}{T}{c} & \emph{Expressions} \\
    \end{array}
    \]
  \begin{flushright}
    \fbox{Expression typing: $\Gamma \vdash e \prdcon T$}
  \end{flushright}
  \vspace{-.4em}
  \begin{prooftree}
    \AxiomC{$\Gamma, x \prdcon T \vdash c \cmd$}
    \RightLabel{\textsc{T-}$\mu$}
    \UnaryInfC{$\Gamma \vdash \muqabs{x}{T}{c} \overset{\flip{o}}{:} T$}
  \end{prooftree}
  \vspace{.2em}
  \begin{flushright}
    \fbox{Reduction: $c_1 \reducesto c_2$}
  \end{flushright}

  \begin{minipage}{0.45\textwidth}
    \begin{prooftree}
      \AxiomC{$\substitutable(e)$}
      \RightLabel{\textsc{R-}$\mu_1$}
      \UnaryInfC{$\muprdabs{x}{T}{c} \gg e\ \,\reducesto\ c [e/x]$}
    \end{prooftree}
  \end{minipage}
  \begin{minipage}{0.45\textwidth}
    \begin{prooftree}
      \AxiomC{$\substitutable(e)$}
      \RightLabel{\textsc{R-}$\mu_2$}
      \UnaryInfC{$e \gg \muconabs{x}{T}{c}\ \,\reducesto\ c [e/x]$}
    \end{prooftree}
  \end{minipage}
  \caption{Syntax, typing and evaluation rules for $\mu$ abstractions. $c [e/x]$ denotes command $c$ where all occurrences of variable $x$ have been replaced by expression $e$.}
  \Description{Syntax, typing and evaluation rules for $\mu$ abstractions.}
\label{fig:muabstraction}
\end{figure}
In order to guarantee confluence we therefore need the $\substitutable(e)$ predicate which we introduced in \cref{sec:formalization}.
If we know that $\substitutable\bigl(\muprdabs{x}{T}{c_1}\bigr)$ and $\substitutable\bigl(\muconabs{y}{T}{c_2}\bigr)$ never hold simultaneously for the same type T, then the two evaluation rules from \cref{fig:muabstraction} don't overlap.

A sensible definition of $\substitutable(e)$ will therefore always determine which $\mu$ abstraction to evaluate first.
This corresponds precisely to specifying an \emph{evaluation order}.
In the following subsection we will therefore discuss some possible evaluation strategies.

\subsection{Specifying evaluation strategies}
\label{subsec:evaluationorder:strategies}

\begin{figure*}[htbp]
  \begin{subfigure}[b]{0.9\textwidth}
    \begin{equation*}
      \substitutable(\mathcal{X}\sigma) \quad
      \substitutable\bigl(\match{p}\, T\, \lbrace\, \ldots \rbrace\bigr) \quad
      \substitutable(x)
    \end{equation*}
    \caption{Common rules}
    \Description{Common rules}
    \label{fig:evalstrategies:common}
  \end{subfigure}
  \vspace{0.2cm}

  \begin{minipage}[b]{0.3\textwidth}
    \begin{subfigure}[b]{0.9\textwidth}
      $\substitutable\bigl(\muconabs{x}{T}{c}\bigr)$
      \caption{Global \cbv}
      \Description{Global call-by-value}
      \label{fig:evalstrategies:cbv}
    \end{subfigure}
    \vskip\baselineskip

    \begin{subfigure}[b]{0.9\textwidth}
      $\substitutable\bigl(\muprdabs{x}{T}{c}\bigr)$
      \caption{Global \cbn}
      \Description{Global call-by-name}
      \label{fig:evalstrategies:cbn}
    \end{subfigure}
  \end{minipage}
  \begin{subfigure}[b]{0.3\textwidth}
    If $T$ has polarity \textbf{data}:
    \vspace{-0.2cm}
    \begin{equation*}
      \substitutable\bigl(\muconabs{x}{T}{c}\bigr)
    \end{equation*}
    If $T$ has polarity \textbf{codata}:
    \vspace{-0.2cm}
    \begin{equation*}
      \substitutable\bigl(\muprdabs{x}{T}{c}\bigr)
    \end{equation*}
    \caption{Polar}
    \Description{Polar evaluation order}
    \label{fig:evalstrategies:polar}
  \end{subfigure}
  \begin{subfigure}[b]{0.3\textwidth}
    If $T$ is a \cbv\ type:
   \vspace{-0.2cm}
    \begin{equation*}
      \substitutable\bigl(\muconabs{x}{T}{c}\bigr)
    \end{equation*}
    If $T$ is a \cbn\ type:
    \vspace{-0.2cm}
    \begin{equation*}
      \substitutable\bigl(\muprdabs{x}{T}{c}\bigr)
    \end{equation*}
    \caption{Nominal}
    \Description{Nominal evaluation order}
    \label{fig:evalstrategies:nominal}
  \end{subfigure}
  \caption{Evaluation orders.}
  \Description{Evaluation orders.}
  \label{fig:evalstrategies}
\end{figure*}

First, let us specify precisely what we mean by ``call-by-value'' and ``call-by-name'' in this data-codata system.
An evaluation strategy is characterized by what kind of expressions we are allowed to substitute for a variable.
For example, call-by-value lambda calculus prohibits the substitution of non-canonical terms such as $1+1$ for $x$ in the redex $(\lambda x.t) (1+1)$.
We can directly translate this example to the symmetric calculus using the functions from \cref{fig:natNumbers:data}.
The term $1 + 1$ corresponds to $\muprdabs{k}{\mathbb{N}}{1 \gg \texttt{add}(1,k)}$.
We don't allow this term to be substituted for a producer variable in a reduction step if we follow a call-by-value reduction strategy.
Generalizing from this concrete example, call-by-value disallows the substitution of any $\mu$-abstractions for producer variables in a reduction step.

We have listed several possible evaluation strategies in \cref{fig:evalstrategies}.
An evaluation strategy is fully specified by the rules that govern the $\substitutable(e)$ predicate.
All evaluation strategies share the rules of \cref{fig:evalstrategies:common}.
These rules state that all function calls, xtors and local matches are substitutable.
Moreover, variables are also substitutable.
The reasoning behind the latter is as follows:
Since expressions are only substituted into other expressions when they appear at the top-level of a command that is being reduced, they cannot be variables, as the command would not be closed otherwise.
Thus, they must have been replaced with some other term in an earlier substitution step and this term must have been itself substitutable for this substitution to occur.

One way to specify an evaluating strategy is to prescribe a \emph{global} evaluation order.
For example, OCaml chose call-by-value, whereas Haskell chose call-by-name.\footnote{We ignore the difference between call-by-name and call-by-need in this article.}
In our formalism, this corresponds to the definitions given in \cref{fig:evalstrategies:cbv} and \cref{fig:evalstrategies:cbn}.
Specifying a global evaluation order works, but has the disadvantage already indicated in the introduction: the validity of $\eta$ is restricted to data types (under call-by-value) or codata types (under call-by-name).

A better alternative is to specify evaluation order per-type; that is, whether a $\mu$ abstraction is substitutable depends on the polarity of the type $T$.
This corresponds to the rules given in \cref{fig:evalstrategies:polar}, which we call the \emph{polar} evaluation order.
The ``polar'' evaluation strategy is the most natural, since it satisfies the most $\eta$-laws.

Ideally we would want to prescribe the polar evaluation order in our system, but this doesn't interact nicely with de- and refunctionalization, which we discuss in the next subsection.
We therefore need an additional evaluation strategy, the \emph{nominal} strategy given in \cref{fig:evalstrategies:nominal}.
In the nominal strategy every type declares whether to evaluate its redexes by-value or by-name.

\subsection{Extensionality}
\label{subsec:evaluationorder:extensionality}

We now come back to the problem of $\eta$ equalities mentioned in \cref{sec:mainideas}.
Assume that the variables bound in the  $\Gamma_i$ do not occur free in the expression $e$.
We define \etaeq\ as the congruence on terms generated by the following two equations.
\begin{align*}
  \match{\data}\, T\ \{ \overline{\X\Delta} \Rightarrow \overline{\X \idmorphism{\Gamma}} \gg e \} &\etaeq e \tag{$\eta_\data$}
\label{eq:eta-data}\\
  \match{\codata}\, T\ \{ \overline{\X\Delta} \Rightarrow e \gg \overline{\X \idmorphism{\Gamma}}  \} &\etaeq e \tag{$\eta_\codata$}
\label{eq:eta-codata}
\end{align*}

In \cref{subsec:mainideas:evaluationorder} we presented an example which showed that the validity of $\eta$ equality depends on the evaluation order of the language.
For the polar evaluation order as defined above, all $\eta$ equalities are valid.

\begin{lemma}
  Assuming the polar evaluation order from \cref{fig:evalstrategies:polar}, we have for any expression $e$ and $e'$:
  \begin{align*}
    (\substitutable(e) \wedge e \etaeq e') \Rightarrow \substitutable(e')
  \end{align*}
\end{lemma}

Since $\etaeq$ is defined as a congruence generated by \cref{eq:eta-data} and \cref{eq:eta-codata}, the following theorem needs the reflexive-transitive closure $\reducesto^\ast$ of the reduction relation $\reducesto$ since multiple expansions might have occurred within one use of $\etaeq$.

\begin{theorem}
  Assuming the polar evaluation order from \cref{fig:evalstrategies:polar}, let $c_1, c_1', c_2$ be commands s.t. $c_1 \reducesto c_2$ and $c_1 \etaeq c_1'$.
  Then there exists a command $c_2'$, s.t. $c_2 \etaeq c_2'$ and $ c_1' \reducesto^{\ast} c_2'$.
\end{theorem}

\section{Changing Evaluation Order}
\label{sec:xtorizationTwo}
In this section, we will define the missing two transformations which change the evaluation strategy of a type, without changing its polarity (data or codata).
That is, we show how to implement the vertical transformations of the diagram introduced in \cref{sec:mainideas}.
\begin{center}
  \begin{tikzcd}
    \cbv\ \data \ar[rr, shift left=0.5ex, gray, draw=gray, "\refunctionalize"] \ar[d, shift left=0.75ex, "\cbvtocbn"] & &
    \cbv\ \codata \ar[ll, shift left=0.5ex, gray, draw=gray, "\defunctionalize"] \ar[d, shift left=0.75ex, "\cbvtocbn"] \\
    \cbn\ \data \ar[rr, shift left=0.5ex, gray, draw=gray, "\refunctionalize"] \ar[u, shift left=0.75ex, "\cbntocbv"] & &
    \cbn\ \codata \ar[ll, shift left=0.5ex, gray, draw=gray, "\defunctionalize"] \ar[u, shift left=0.75ex, "\cbntocbv"]
  \end{tikzcd}
\end{center}
The full refunctionalization from call-by-value data to call-by-name codata is then just $\refunctionalize$, extended to account for $\mu$,
followed by $\cbvtocbn$ (or the other way around, which makes no difference), and similarly for defunctionalization.
The $\cbvtocbn$ and $\cbntocbv$ transformations use so-called \emph{shifts}, which we introduce in \cref{subsec:evaluationorder:shifts}, to embed cbn into cbv types or vice versa.
When we apply defunctionalization and then refunctionalization (or vice versa),
we use $\cbvtocbn$ after $\cbntocbv$ and thus introduce \emph{double-shift} artifacts
(thanks to $\defunctionalize$ and $\refunctionalize$ being mutually inverse,
apart from that aspect, we do get back to the original program).
However, we show how these artifacts can be removed so that the two transformations
are indeed mutually inverse.

\subsection{\Derefunc\ and evaluation order}
\label{subsec:evaluationorder:xfunc}

Extending the transformations $\xfunctionalize$ of \cref{sec:xtorization} to cover $\mu$ abstractions is straight-forward:
Since evaluation order is not changed in these transformations, we can just keep all $\mu$ abstractions within the program unchanged, i.e.\ the transformation on terms still remains the identity function.
Furthermore, the rest of the transformation remains a simple matrix transposition, as before.
However, as explained in \cref{subsec:mainideas:derefunceval}, the resulting transformation is unsatisfying on its own, since it does not allow the input and output to both have the transformed type in its polar, i.e. natural evaluation order.
Therefore, we introduce \emph{nominal} evaluation order and the $\cbvtocbn$ and $\cbntocbv$ transformations.

Specifically, there are two observations that we can make:
Firstly, this extended transformation only preserves typechecking if we keep the evaluation order of the type $T$ constant.
Otherwise we might violate the additional premise in the \textsc{Subst}$_2$ rule.
Secondly, even if the transformation were type-preserving, it would not be semantics-preserving, since we would now disambiguate the critical pair differently.
Thus, the central problem that we tackle in this section is:
\textbf{We cannot combine the simple \derefunc\ approach described above with polar evaluation order.}

Previously, user-defined types could be either data or codata types.
We now extend this scheme and parameterize user-defined types by their evaluation order,
written with prefix \cbv\ or \cbn\, e.g.
{\footnotesize
  \begin{align*}
    \cbv\ \data\ \mathbf{type}\ \N\ \{\ \texttt{Zero}()\ ;\ \texttt{Suc}(x \prd \N)\ \}.
  \end{align*}
}%
The evaluation order is now induced by this parameter of the type,
as formally defined by the two rules for the $\substitutable({-})$ judgement
for \emph{nominal} evaluation order (\cref{fig:evalstrategies:nominal}).

As sketched above, we can now factorize the transformations between a \cbv\ data type and a \cbn\ codata type into two steps,
the first of which changes the polarity of the type without changing its evaluation order
($\refunctionalize$ and $\defunctionalize$ presented in \cref{sec:xtorization}, trivially extended to $\mu$).
In order to define the second part of this algorithm, we have to introduce shift types, which we do in the next subsection.

\subsection{Shift Types}
\label{subsec:evaluationorder:shifts}

The shift types, briefly introduced in \cref{subsec:mainideas:summary}, are ordinary user-defined types in this system; the type system does not need to be extended for them, and the normal evaluation and typing rules presented above apply.
We allow ourselves to write the definitions of the shifts parameterized by a type, even though our formalism doesn't allow this.
This is not essential; we could alternatively add one shift type to the program for every type that we want to shift.
The definitions of the shift types are:
{\footnotesize
\begin{align*}
  &\cbv\ \data\ \mathbf{type}\ \cbvshift{T}\ \{\ \cbvterm{x \prd T}\ \} \\
  &\cbn\ \codata\ \mathbf{type}\ \cbnshift{T}\ \{\ \cbnterm{x \con T}\ \}
\end{align*}
}

As is apparent from their definition, these types do not change the logical meaning of the type they shift.
Their effect is restricted to the evaluation order of the program, in every other respect they behave like an identity wrapper.
This is why the following helper function $\S^o_s$ allows to embed any producer or consumer expression in the corresponding shifted type:
\begin{align*}
  \S^{\mathbf{prd}}_\cbv(e) &\coloneq \cbvterm{e} \\
  \S^{\mathbf{prd}}_\cbn(e) &\coloneq \mathbf{match}_\codata\,\cbnshift{T}\ \{ \cbnterm{x} \Rightarrow e \gg x\} \\
  \S^{\mathbf{con}}_\cbv(e) &\coloneq \mathbf{match}_\data\,\cbvshift{T}\ \{ \cbvterm{x} \Rightarrow x \gg e \} \\
  \S^{\mathbf{con}}_\cbn(e) &\coloneq \cbnterm{e}
\end{align*}
That is, the following rule is derivable for all $o, s, e$ and $T$:
\begin{prooftree}
  \AxiomC{$\Gamma \vdash e \overset{o}{:} T$}
  \UnaryInfC{$\Gamma \vdash \S^o_s(e) \overset{o}{:}\ \cbxshift{T}$}
\end{prooftree}

We will now give some examples for how shift types can be used.
The \cbv\ data type $\mathbb{N}$ of natural numbers can be wrapped as $\cbnshift{\mathbb{N}}$ to obtain a codata type of \emph{delayed} natural numbers.
A list of ordinary natural numbers cannot contain the unevaluated expression $\muprdabs{k}{\mathbb{N}}{(1 \gg \texttt{add}(1,k))}$, since this expression is not substitutable and can therefore (according to rule \textsc{T-Subst}$_2$) not occur in a substitution.
However, a list of delayed natural numbers can contain the equivalent expression $\mathbf{match}_\codata\,\cbnshift{T}\ \{ \cbnterm{k} \Rightarrow 1 \gg \texttt{add}(1,k) \}$.
This allows a very fine-grained control over evaluation order in the types.
In a program which contains \cbv\ data type definitions of natural numbers and lists, as well as a function type, three different kinds of functions which expect a list of natural numbers can be distinguished.
A function with argument of type $\texttt{List}\ \N$ evaluates its arguments to a list of fully evaluated natural numbers.
If the argument type is $\texttt{List}\ (\cbnshift \N)$, then the spine of the list has to be fully evaluated before it is substituted in the body of the function, but the elements of the list might still be unevaluated.
Thirdly, if an argument of type $\cbnshift (\texttt{List}\ \N)$ is expected, the argument is passed by-name to the body of the function.

\subsection{Changing the evaluation order}
\label{subsec:evaluationorder:changingstrategy}

\begin{figure*}[htbp]
  \begin{subfigure}[b]{\textwidth}
    \begin{minipage}{.55\linewidth}
      \begin{align*}
        \evaltrans{s'\ p\ &\mathbf{type}\ T'\ \{ \overline{\X\,\Delta} \}\ \mathbf{with}\ \overline{\Y \Pi := e}}
        \coloneq\\[-.3em]
        & s''\ p\ \mathbf{type}\ T'\ \{ \overline{\X}\ \evaltrans{\overline{\Delta}} \}\ \mathbf{with}\ \overline{\Y \evaltrans{\Pi} := \evaltrans{e}}
      \end{align*}
    \end{minipage}
    \begin{minipage}{.35\linewidth}
      \centering \hfill where $ s'' \coloneq \begin{cases} s' &\quad (T \neq T') \\ s\ &\quad (T = T') \end{cases} $
    \end{minipage}
    \caption{Type declarations.}
    \Description{Type declarations.}
    \label{fig:evalchange:type}
  \end{subfigure}

  \begin{subfigure}[b]{0.9\textwidth}
    \begin{align*}
      \evaltrans{\diamond} &\coloneq \diamond &
      \evaltrans{\Gamma, x \overset{o}{:} T'} &\coloneq
      \begin{cases}
          \evaltrans{\Gamma}, x \overset{o}{:} T' & (T \neq T')\\
          \evaltrans{\Gamma}, x \overset{o}{:}\  \cbxshiftinv{T} & (T = T')
      \end{cases} \\
      \evaltrans{()} &\coloneq () &
      \evaltrans{(\sigma, e} &\coloneq (\evaltrans{\sigma},\evaltrans{e})
    \end{align*}
    \caption{Contexts and Substitutions.}
    \label{fig:evalchange:context}
  \end{subfigure}

  \begin{subfigure}[b]{0.9\textwidth}
    \begin{equation*}
      \evaltrans{\Done} \coloneq \Done \qquad
      \evaltrans{e_1 \gg e_2} \coloneq  \evaltrans{e_1} \gg \evaltrans{e_2}
    \end{equation*}
    \caption{Commands.}
    \label{fig:evalchange:commands}
  \end{subfigure}

  \begin{subfigure}[b]{0.9\textwidth}
    \begin{align*}
      \evaltrans{x} & \coloneq x \\
      \evaltrans{\X \sigma} & \coloneq
      \begin{cases}
        \X \, \evaltrans{\sigma}                    & \X \notin \text{Funs}(T)\ \text{and}\ \X \notin \text{Xtors}(T)\\
        \S^{\valfun{p}}_{\flip{s}}(\X \  \evaltrans{\sigma})   & \X \in \text{Xtors}(T)\ \text{and Polarity}(T) = p\\
        \S^{\cntfun{p}}_{\flip{s}}(\X \  \evaltrans{\sigma})   & \X \in \text{Funs}(T)\ \text{and Polarity}(T) = p\\
      \end{cases} \\
      \evaltrans{\match{p}\, T'\,\{\overline{\X \Delta} \Rightarrow \overline{c}\}} & \coloneq
      \begin{cases}
        \match{p}\,T'\, \{\overline{\X \Delta} \Rightarrow \evaltrans{\overline{c}}\} & (T \neq T')\\
        \S^{\cntfun{p}}_{\flip{s}}(\match{p}\,T\, \{\overline{\X \Delta} \Rightarrow \evaltrans{\overline{c}}\}) & (T = T')
      \end{cases} \\
      \evaltrans{\muqabs{x}{T'}{c}} & \coloneq
      \begin{cases}
        \muqabs{x}{T'}{\evaltrans{c}} & (T \neq T') \\
        \muqabs{x}{\ \cbxshiftinv{T}}{\evaltrans{c}} & (T = T')
      \end{cases}
    \end{align*}
    \caption{Expressions.}
    \label{fig:evalchange:expressions}
  \end{subfigure}
  \caption{Changing evaluation order of a type.}
  \label{fig:evalchange}
\end{figure*}

\cref{fig:evalchange} defines the transformation \evaltransname\!, which changes the evaluation strategy specified for type $T$ to the evaluation strategy $s$, for the different syntactic entities.
For type declarations (\cref{fig:evalchange:type}), $\evaltransname$ of course changes the specified evaluation order of $T$ to $s$, and in all declarations applies $\evaltransname$ to the contexts in the signatures.
Note that we assume that the transformation is only applied to applied to types where it makes sense, i.e.\ $\evaltransnametemplate{T}{\cbv}$ is only applied if $T$ is a \cbn\ type and $\evaltransnametemplate{T}{\cbn}$ only if $T$ is a \cbv\ type. 
If used this way, note that in the case of $T = T'$, we also have $s = \flip{s'}$.

In contexts (\cref{fig:evalchange:context}), $\evaltransname$ replaces $T$ with the relevant shifted type to retain the original evaluation order $s$.

Accordingly, in the appropriate places in expressions (\cref{fig:evalchange:expressions}) the type of a $\mu$ has to be changed to the corresponding shifted type,
and a $\mathbf{match}$ or a call to \texttt{CBV} or \texttt{CBN} has to be inserted to wrap expressions of type $T$, resulting in the corresponding shifted type.
Thus, overall, the evaluation order specified for $T$ is changed while retaining the operational semantics of the program.%
\footnote{Up to reduction of administrative redexes caused by the wrapping.}

\subsection{Converting from call-by-value to call-by-name and back: removing double-shifts}

When using the transformation twice on a program $P$ w.r.t. a type T, the resulting program, i.e. $\doubletrans{P}$, contains double-shift artifacts, which can be removed in order to obtain the original program (with type $T$ replacing type $\cbxshiftinv{\cbxshift{T}}$).
Expressions of type $\cbvshift{\cbnshift{T}}$ which contain such artifacts are either
{\footnotesize
  \begin{align*}
    \cbvterm{\match{p} \cbnshift{T}\  \{ \cbnterm{k} \Rightarrow e \gg k \}}
  \end{align*}
}%
or
{\footnotesize
  \begin{align*}
    \muprdabs{x}{\cbvshift{\cbnshift{T}}}{c}.
  \end{align*}
}%
These can respectively be replaced by the result of recursively replacing subexpressions
in the expressions $e$ and $\muprdabs{x}{T}{c}$, obtaining expressions of type $T$.
For a type $\cbnshift{\cbvshift{T}}$, we proceed analogously.
By a simple structural induction, one can show that the just described transformation, which we call $\mathcal{A}^T$, indeed goes from $\doubletrans{P}$ to $P$.

Combining this result with \cref{theorem:progsinverse}, we have shown that walking from one corner of our diagram to another and back (and eliminating double-shifts) leads back to the original program.

\begin{theorem}[Mutual Inverse (overall)]
  \label{theorem:progsinverse-overall}
  For every well-formed program $P$ and data type $T$ in $P$:
  \begin{align*}
    \mathcal{A}^T(\evaltransinv{\xfunctionalizeinv^{T}(\evaltrans{\xfunctionalize^{T}(P)})}) = P
  \end{align*}
\end{theorem}

\section{Related Work}
\label{sec:relatedwork}

There are two separate strands of related work: the development of the theory of defunctionalization and refunctionalization on the one hand, and the development of symmetric calculi on the other hand.

\paragraph{Defunctionalization and refunctionalization}
\Derefunc\ of the function type have a long history, of which we only cite the seminal papers \citep{reynolds72definitional,danvy01defunctionalization,danvy09refunctionalization}.
The generalization of defunctionalization from functions to arbitrary codata types was described by \citet{uroboro2015} for a simply typed system without local lambda abstractions or local pattern matches.
That the defunctionalization of polymorphic functions requires GADTs was first observed by \citet{pottier06polymorphic};
that the generalization to data and codata types then also requires GAcoDTs has been described by \citet{ostermann2018dualizing}.
How to treat local pattern and copattern matches in such a way as to preserve invertibility of defunctionalization and refunctionalization has been described by \citet{Uroboro2019}.
The novel aspect of the present work is that it presents completely symmetric data and codata types, and that it also considers the interaction with evaluation order.

\paragraph{Symmetric calculi}
Our core calculus was strongly influenced by the ideas of \citeauthor{zeilberger2008unity}'s Calculus of Unity (\CU) \citep{zeilberger2008unity}.
In contrast to our system, \CU\ does not provide direct means for user-defined data types
(though a related system does have a similar mechanism \citep{zeilberger2008focusing}).
Furthermore, its treatment of pattern-matches on recursive types is not syntactical: a pattern match on $\mathbb{N}$ contains an infinite amount of cases, similar to the $\omega$-rule \cite{hilbert1931grundlegung} in formal systems of arithmetic.
By contrast, our system (which is hence easily implementable) only allows finite matches,
which also have to be shallow anyway, in order to facilitate de- and refunctionalization.

The other important heritage is the ongoing quest for proof-theoretically well-behaved term assignment systems for sequent calculus.
The \lambdamutildemu\ calculus of \citet{curien00duality} is such a system with a fixed set of types; its authors discuss the problems of confluence and evaluation order, but do not consider data and codata types in their general form.
The ``codata--data'' calculus (\DualCalculus) of \citet{downen2019compiling} is an extension of \lambdamutildemu\ with user-defined data and codata types, polymorphism and higher kinds.
They discuss extensively the relation between evaluation order and the polarity of types (data or codata), but do not consider algorithms for switching these properties.
Their data and codata types are symmetric, but they do not exploit this fact in their formalization; instead they present separate rules for data and codata types.

Another well-known symmetric calculus is the dual calculus of \citet{wadler03call}.
This calculus uses the same judgements as our system and \lambdamutildemu, and contains both $\mu$ and $\tilde\mu$ constructs.
The types of that system are of ambiguous polarity; for example, the product type is defined by a positive introduction rule (pairing) and negative elimination rules (first and second projection).
\citeauthor{wadler03call} also discusses De Morgan duality.
He defines a dualizing operation $-^\circ$ which behaves as an involution on both types and terms, and relates proofs of a type $T$ to refutations of $T^\circ$ (in our notation: $\Gamma \vdash t \prd T \Leftrightarrow \Gamma^\circ \vdash t^\circ \con T^\circ$).
It is possible to define the same operation in our system.
Taking polarity into account, this relates proofs (refutations) of a product data type $\otimes$, i.e.\ a data type with one ``pair'' constructor, to refutations (proofs) of sum codata type $\parr$, a codata type with one ``case'' destructor.
The proofs and refutations for the sum data type $\oplus$, a data type with two constructors ``inl'' and ``inr'', and products codata type $\&$, a codata type with two destructors ``outl'' and ``outr'', are similarly related.

\section{Conclusion}
\label{sec:conclusion}
We have presented a system with completely symmetric data and codata types.
We have also presented transformations which transform cbv data into cbn codata, and vice versa; this transformation has been factored into one transformation which changes the polarity and one transformation which changes the evaluation order.
This has revealed that evaluation order and polarity of a type, even though related, are best treated separately when considering the conversion between data and codata.
In particular, codata types have more to offer than just the representation of infinite data, they are an essential ingredient in getting the evaluation strategy of a language right.
We believe that future programming languages should support both data and codata types, and that their interaction with evaluation order should follow a principled design.
We think that this article has presented one such approach.

\bibliography{bib}

\appendix

\end{document}